\documentclass[conference,compsoc,10pt]{IEEEtran}

\usepackage{graphicx}
\usepackage{amsfonts}
\usepackage{amsmath}  
\usepackage{amsthm}
\usepackage{amssymb}
\usepackage{relsize}
\usepackage{mathtools}
\usepackage[applemac]{inputenc} 
\usepackage{hyperref}
\usepackage[nocompress]{cite}

\usepackage[usenames,dvipsnames]{xcolor}
\usepackage{tikz}
\usetikzlibrary{arrows,shapes,snakes,backgrounds,positioning,circuits.ee.IEC}
\usepackage{tikzfig}

\usepackage[nodisplayskipstretch]{setspace}

	\newcounter{theorem_c}
	\numberwithin{theorem_c}{section}

	\theoremstyle{plain}
	\newtheorem{theorem}[theorem_c]{Theorem}
	
	\newtheorem{lemma}[theorem_c]{Lemma}
	\newtheorem{corollary}[theorem_c]{Corollary}

	\theoremstyle{definition}
	\newtheorem{definition}[theorem_c]{Definition}
	
	\newtheorem{remark}[theorem_c]{Remark}

	\newcommand\numberthis{\addtocounter{equation}{1}\tag{\theequation}}

	\newcommand{\emptyArg}{\,\underline{\hspace{6px}}\,}
	\newcommand{\inlineQuote}[1]{\textquotedblleft #1\textquotedblright}

	\newcommand{\integers}{\mathbb{Z}}
	
	\newcommand{\reals}{\mathbb{R}}

	\newcommand{\integersMod}[1]{\mathbb{Z}_{#1}}
	\newcommand{\modclass}[2]{#1 \; (\text{mod } #2)}

	\newcommand{\eqdef}{\stackrel{def}{=}}

		\newcommand{\isom}{\cong}
		\newcommand{\epim}{\twoheadrightarrow}
		\newcommand{\monom}{\rightarrowtail}

		\newcommand{\tensor}{\otimes}
		\newcommand{\tensorUnit}{I}
		
		\newcommand{\id}[1]{id_{#1}}

		\newcommand{\Automs}[2]{\operatorname{Aut}_{\,#1}\left[#2\right]}
		\newcommand{\Isoms}[3]{\operatorname{Iso}_{\,#1}\left[#2,#3\right]}

		\newcommand{\States}[2]{\operatorname{States}_{#1}[#2]}

			\newcommand{\fdHilbCategory}{\operatorname{fdHilb}}

			\newcommand{\AbCategory}{\operatorname{Ab}}

			\newcommand{\CategoryC}{\mathcal{C}}

		\newcommand{\monad}{T}
		\newcommand{\unit}[1]{\eta_{#1}}
		\newcommand{\mult}[1]{\mu_{#1}}
			 
			\newcommand{\timeobj}{\mathbb{T}}

			\newcommand{\freehistory}[1]{\operatorname{freehist}_{#1}}

			\newcommand{\concretehistory}[2]{\operatorname{hist}^{#1}_{#2}} 
			\newcommand{\ConcreteHistories}[2]{\operatorname{Hists}_{#1}[#2]}
			
			\newcommand{\timeinversion}[1]{i_{#1}}
			\newcommand{\foliate}[1]{\operatorname{foliate}_{#1}}

			\newcommand{\timeunit}{\unit{}}
			\newcommand{\timemult}{\mult{}}

			\newcommand{\timeinverse}{\timeinversion{}}

	\newcommand{\ket}[1]{\vert #1 \rangle}
	\newcommand{\bra}[1]{\langle #1 \vert}
	
	\newcommand{\SpaceH}{\mathcal{H}}

	\newcommand{\SpaceG}{\mathcal{G}}
	
	\newcommand{\UnitaryOps}[1]{\operatorname{U}\left[#1\right]}
	
	\newcommand{\Dim}[1]{\operatorname{dim}#1}
	
	\newcommand{\LtwoSym}{\operatorname{L}^2}
	\newcommand{\Ltwo}[1]{\LtwoSym[#1]}
	\newcommand{\Irreps}[1]{\operatorname{Irr}[#1]}
	\newcommand{\Norm}[2]{|| #2 ||_{#1}}
	\newcommand{\Bounded}[1]{\operatorname{B}\left[#1\right]}

		\tikzstyle{trivial}=
			[
			draw=none,
			fill=none,
			thick
			]
		\tikzstyle{filler}=
			[
			draw=black,
			fill=black,
			thick,
			inner sep=0.02pt,
			minimum size=0mm
			] 
		\tikzstyle{system}=
			[
			fill=none
			]
		\tikzstyle{morphism}=
			[
			inner sep=0mm,
			minimum height=6mm,
			minimum width=10mm,
			draw=black!75,
			trapezium, trapezium left angle=90
			]
		\tikzstyle{spider}=
			[
			inner sep=0mm,
			minimum size=6mm,
			draw=black!75,
			circle
			]
		\tikzstyle{state}=
			[
			inner sep=0mm,
			minimum height=6mm,
			minimum width=10mm,
			draw=black!75,
			regular polygon, regular polygon sides=3, shape border rotate=180
			]
		\tikzstyle{stateNSC}=
			[
			inner sep=0mm,
			minimum height=6mm,
			minimum width=10mm,
			draw=black!75,
			regular polygon, regular polygon sides=3, shape border rotate=180
			]
		\tikzstyle{effect}=
			[
			inner sep=0mm,
			minimum height=6mm,
			minimum width=10mm,
			draw=black!75,
			regular polygon, regular polygon sides=3
			]
		\tikzstyle{effectNSC}=
			[
			inner sep=0mm,
			minimum height=6mm,
			minimum width=10mm,
			draw=black!75,
			regular polygon, regular polygon sides=3
			]
		\tikzstyle{identity}=
			[
			inner sep=0pt,
			minimum size=0.1pt
			]
		\tikzstyle{boundary}=
			[
			draw=none,
			fill=none,
			inner sep=0pt,
			minimum size=5mm
			]

		\tikzstyle{rectangle}=
			[
			draw=black,
			fill=none 
			] 
		\tikzstyle{timespace}=
			[
			draw=orange!75,
			fill=orange!25,
			thick
			] 	
		\tikzstyle{freqspace}=
			[
			draw=green!75,
			fill=green!25,
			thick
			] 	
		\tikzstyle{algebra}=
			[
			draw=red!75,
			fill=red!25,
			thick
			] 

		\newcommand{\SymFontShift}{0pt}

		\newcommand{\groupStructColour}{Red}
		\newcommand{\classicalStructColour}{YellowGreen}

		\newcommand{\internalgroupStructColour}{Purple}
		\newcommand{\internalclassicalStructColour}{Cyan}

		\newcommand{\internaltimestructure}{(\raisebox{\SymFontShift}{\hbox{\input{modules/symbols/internaltimediagSym.tex}}}\!,\raisebox{\SymFontShift}{\hbox{\input{modules/symbols/internaltrivialcharSym.tex}}}\!,\raisebox{\SymFontShift}{\hbox{\input{modules/symbols/internaltimematchSym.tex}}}\!,\raisebox{\SymFontShift}{\hbox{\input{modules/symbols/internaltimematchunitSym.tex}}}\!)}
		\newcommand{\internaltimecomonoid}{\raisebox{\SymFontShift}{\hbox{\input{modules/symbols/internaltimediagSym.tex}}}\!,\raisebox{\SymFontShift}{\hbox{\input{modules/symbols/internaltrivialcharSym.tex}}}\!}

		\newcommand{\internalenergystructure}{(\raisebox{\SymFontShift}{\hbox{\input{modules/symbols/internaltimemultSym.tex}}}\!,\raisebox{\SymFontShift}{\hbox{\input{modules/symbols/internaltimeunitSym.tex}}}\!,\raisebox{\SymFontShift}{\hbox{\input{modules/symbols/internaltimecomultSym.tex}}}\!,\raisebox{\SymFontShift}{\hbox{\input{modules/symbols/internaltimecounitSym.tex}}}\!)}
		\newcommand{\internalenergycomonoid}{(\raisebox{\SymFontShift}{\hbox{\input{modules/symbols/internaltimecomultSym.tex}}}\!,\raisebox{\SymFontShift}{\hbox{\input{modules/symbols/internaltimecounitSym.tex}}}\!)}
		\newcommand{\internalenergymonoid}{\raisebox{\SymFontShift}{\hbox{\input{modules/symbols/internaltimemultSym.tex}}}\!,\raisebox{\SymFontShift}{\hbox{\input{modules/symbols/internaltimeunitSym.tex}}}\!}

		\newcommand{\groupElementSymName}{}

		\newcommand{\repElementSymName}{}

		\newcommand{\tikzEquationsScale}{1.5}

\tikzstyle{env}=[copoint,regular polygon rotate=0,minimum width=0.2cm, fill=black]

\tikzstyle{probs}=[shape=semicircle,fill=white,draw=black,shape border rotate=180,minimum width=1.2cm]

\tikzstyle{every picture}=[baseline=-0.25em,scale=0.5]
\tikzstyle{dotpic}=[] 
\tikzstyle{diredges}=[every to/.style={diredge}]
\tikzstyle{math matrix}=[matrix of math nodes,left delimiter=(,right delimiter=),inner sep=2pt,column sep=1em,row sep=0.5em,nodes={inner sep=0pt},text height=1.5ex, text depth=0.25ex]

\tikzstyle{inline text}=[text height=1.5ex, text depth=0.25ex,yshift=0.5mm]
\tikzstyle{label}=[font=\footnotesize,text height=1.5ex, text depth=0.25ex,yshift=0.5mm]
\tikzstyle{left label}=[label,anchor=east,xshift=1.5mm]
\tikzstyle{right label}=[label,anchor=west,xshift=-1.5mm]

\tikzstyle{braceedge}=[decorate,decoration={brace,amplitude=2mm,raise=-1mm}]
\tikzstyle{small braceedge}=[decorate,decoration={brace,amplitude=1mm,raise=-1mm}]

\tikzstyle{doubled}=[line width=1.6pt] 
\tikzstyle{boldedge}=[doubled,shorten <=-0.17mm,shorten >=-0.17mm]
\tikzstyle{boldedgegray}=[doubled,gray,shorten <=-0.17mm,shorten >=-0.17mm]

\tikzstyle{semidoubled}=[line width=1.4pt] 
\tikzstyle{semiboldedgegray}=[semidoubled,gray,shorten <=-0.17mm,shorten >=-0.17mm]

\tikzstyle{boldedgedashed}=[very thick,dashed,shorten <=-0.17mm,shorten >=-0.17mm]
\tikzstyle{vboldedgedashed}=[doubled,dashed,shorten <=-0.17mm,shorten >=-0.17mm]
\tikzstyle{left hook arrow}=[left hook-latex]
\tikzstyle{right hook arrow}=[right hook-latex]
\tikzstyle{sembracket}=[line width=0.5pt,shorten <=-0.07mm,shorten >=-0.07mm]

\tikzstyle{causal edge}=[->,thick,gray]
\tikzstyle{causal nondir}=[thick,gray]
\tikzstyle{timeline}=[thick,gray, dashed]

\tikzstyle{cedge}=[<->,thick,gray!70!white]

\tikzstyle{empty diagram}=[draw=gray!40!white,dashed,shape=rectangle,minimum width=1cm,minimum height=1cm]
\tikzstyle{empty diagram small}=[draw=gray!50!white,dashed,shape=rectangle,minimum width=0.6cm,minimum height=0.5cm]

\tikzstyle{dot}=[inner sep=0mm,minimum width=3mm,minimum height=3mm,draw,shape=circle,text depth=-0.1mm]
\tikzstyle{ddot}=[inner sep=0mm, doubled, minimum width=3.5mm,minimum height=3.5mm,draw,shape=circle]

\tikzstyle{black dot}=[dot,fill=black]
\tikzstyle{white dot}=[dot,fill=white,,text depth=-0.2mm]
\tikzstyle{green dot}=[white dot] 
\tikzstyle{gray dot}=[dot,fill=gray!40!white,,text depth=-0.2mm]
\tikzstyle{red dot}=[gray dot] 

\tikzstyle{black ddot}=[ddot,fill=black]
\tikzstyle{white ddot}=[ddot,fill=white]
\tikzstyle{gray ddot}=[ddot,fill=gray!40!white]

\tikzstyle{gray edge}=[gray!40!white]

\tikzstyle{small dot}=[inner sep=0.5mm,minimum width=0pt,minimum height=0pt,draw,shape=circle]

\tikzstyle{small black dot}=[small dot,fill=black]
\tikzstyle{small white dot}=[small dot,fill=white]
\tikzstyle{small gray dot}=[small dot,fill=gray!40!white]

\tikzstyle{causal dot}=[inner sep=0.4mm,minimum width=0pt,minimum height=0pt,draw=white,shape=circle,fill=gray!40!white]

\tikzstyle{phase dimensions}=[minimum size=5mm,font=\footnotesize,rectangle,rounded corners=2.5mm,inner sep=0.2mm,outer sep=-2mm,text height=1ex, text depth=0.25ex, yshift=0.5mm]
\tikzstyle{dphase dimensions}=[phase dimensions]

\tikzstyle{phase dot}=[dot,phase dimensions]

\tikzstyle{white phase dot}=[dot,fill=white,phase dimensions]
\tikzstyle{white phase ddot}=[ddot,fill=white,dphase dimensions]

\tikzstyle{white rect ddot}=[draw=black,fill=white,doubled,minimum size=5mm,font=\footnotesize,rectangle,rounded corners=2.5mm,inner sep=0.2mm]
\tikzstyle{gray rect ddot}=[draw=black,fill=gray!40!white,doubled,minimum size=6mm,font=\footnotesize,rectangle,rounded corners=3mm]

\tikzstyle{gray phase dot}=[dot,fill=gray!40!white,phase dimensions]
\tikzstyle{gray phase ddot}=[ddot,fill=gray!40!white,dphase dimensions]
\tikzstyle{grey phase dot}=[gray phase dot]
\tikzstyle{grey phase ddot}=[gray phase ddot]

\tikzstyle{cnot}=[fill=white,shape=circle,inner sep=-1.4pt]
\tikzstyle{hadamard}=[square box,inner sep=0 pt,font=\footnotesize,minimum height=4mm,minimum width=4mm]
\tikzstyle{dhadamard}=[hadamard,doubled]
\tikzstyle{antipode}=[white dot,inner sep=0.3mm,font=\footnotesize]

\tikzstyle{scalar}=[diamond,draw,inner sep=0.5pt,font=\small]
\tikzstyle{dscalar}=[diamond,doubled, draw,inner sep=0.5pt,font=\small]

\tikzstyle{small box}=[rectangle,inline text,fill=white,draw,minimum height=5mm,yshift=-0.5mm,minimum width=5mm,font=\small]
\tikzstyle{small gray box}=[small box,fill=gray!30]
\tikzstyle{medium box}=[rectangle,inline text,fill=white,draw,minimum height=5mm,yshift=-0.5mm,minimum width=10mm,font=\small]
\tikzstyle{square box}=[small box] 
\tikzstyle{medium gray box}=[small box,fill=gray!30]
\tikzstyle{semilarge box}=[rectangle,inline text,fill=white,draw,minimum height=5mm,yshift=-0.5mm,minimum width=12.5mm,font=\small]
\tikzstyle{large box}=[rectangle,inline text,fill=white,draw,minimum height=5mm,yshift=-0.5mm,minimum width=15mm,font=\small]
\tikzstyle{large gray box}=[small box,fill=gray!30]

\tikzstyle{gray square point}=[small box,fill=gray!50]

\tikzstyle{dphase box white}=[dbox]
\tikzstyle{dphase box gray}=[dbox,fill=gray!50!white]

\tikzstyle{point}=[regular polygon,regular polygon sides=3,draw,scale=0.75,inner sep=-0.5pt,minimum width=9mm,fill=white,regular polygon rotate=180]
\tikzstyle{copoint}=[regular polygon,regular polygon sides=3,draw,scale=0.75,inner sep=-0.5pt,minimum width=9mm,fill=white]
\tikzstyle{dpoint}=[point,doubled]
\tikzstyle{dcopoint}=[copoint,doubled]

\tikzstyle{wide copoint}=[fill=white,draw,shape=isosceles triangle,shape border rotate=90,isosceles triangle stretches=true,inner sep=0pt,minimum width=1.5cm,minimum height=6.12mm]
\tikzstyle{wide point}=[fill=white,draw,shape=isosceles triangle,shape border rotate=-90,isosceles triangle stretches=true,inner sep=0pt,minimum width=1.5cm,minimum height=6.12mm,yshift=-0.0mm]
\tikzstyle{wide point plus}=[fill=white,draw,shape=isosceles triangle,shape border rotate=-90,isosceles triangle stretches=true,inner sep=0pt,minimum width=1.74cm,minimum height=7mm,yshift=-0.0mm]

\tikzstyle{wide dpoint}=[fill=white,doubled,draw,shape=isosceles triangle,shape border rotate=-90,isosceles triangle stretches=true,inner sep=0pt,minimum width=1.5cm,minimum height=6.12mm,yshift=-0.0mm]

\tikzstyle{tinypoint}=[regular polygon,regular polygon sides=3,draw,scale=0.55,inner sep=-0.15pt,minimum width=6mm,fill=white,regular polygon rotate=180] 

\tikzstyle{white point}=[point]
\tikzstyle{white dpoint}=[dpoint]
\tikzstyle{green point}=[white point] 
\tikzstyle{white copoint}=[copoint]
\tikzstyle{gray point}=[point,fill=gray!40!white]
\tikzstyle{gray dpoint}=[gray point,doubled]
\tikzstyle{red point}=[gray point] 
\tikzstyle{gray copoint}=[copoint,fill=gray!40!white]
\tikzstyle{gray dcopoint}=[gray copoint,doubled]

\tikzstyle{black point}=[point,fill=black]
\tikzstyle{black copoint}=[copoint,fill=black]

\tikzstyle{tiny gray point}=[tinypoint,fill=gray!40!white]

\tikzstyle{diredge}=[->]
\tikzstyle{rdiredge}=[<-]
\tikzstyle{thickdiredge}=[->, very thick]
\tikzstyle{pointer edge}=[->,very thick,gray]
\tikzstyle{pointer edge part}=[very thick,gray]
\tikzstyle{dashed edge}=[dashed]
\tikzstyle{thick dashed edge}=[very thick,dashed]
\tikzstyle{thick gray dashed edge}=[thick dashed edge,gray!40]
\tikzstyle{thick map edge}=[very thick,|->]

\makeatletter
\newcommand{\boxshape}[3]{%
\pgfdeclareshape{#1}{
\inheritsavedanchors[from=rectangle] 
\inheritanchorborder[from=rectangle]
\inheritanchor[from=rectangle]{center}
\inheritanchor[from=rectangle]{north}
\inheritanchor[from=rectangle]{south}
\inheritanchor[from=rectangle]{west}
\inheritanchor[from=rectangle]{east}
\backgroundpath{
\southwest \pgf@xa=\pgf@x \pgf@ya=\pgf@y
\northeast \pgf@xb=\pgf@x \pgf@yb=\pgf@y

\@tempdima=#2
\@tempdimb=#3

\pgfpathmoveto{\pgfpoint{\pgf@xa - 5pt + \@tempdima}{\pgf@ya}}
\pgfpathlineto{\pgfpoint{\pgf@xa - 5pt - \@tempdima}{\pgf@yb}}
\pgfpathlineto{\pgfpoint{\pgf@xb + 5pt + \@tempdimb}{\pgf@yb}}
\pgfpathlineto{\pgfpoint{\pgf@xb + 5pt - \@tempdimb}{\pgf@ya}}
\pgfpathlineto{\pgfpoint{\pgf@xa - 5pt + \@tempdima}{\pgf@ya}}
\pgfpathclose
}
}}

\boxshape{NEbox}{0pt}{5pt}
\boxshape{SEbox}{0pt}{-5pt}
\boxshape{NWbox}{5pt}{0pt}
\boxshape{SWbox}{-5pt}{0pt}
\boxshape{EBox}{-3pt}{3pt}
\boxshape{WBox}{3pt}{-3pt}
\makeatother

\tikzstyle{cloud}=[shape=cloud,draw,minimum width=1.5cm,minimum height=1.5cm]

\tikzstyle{map}=[draw,shape=NEbox,inner sep=2pt,minimum height=6mm,fill=white]
\tikzstyle{dashedmap}=[draw,dashed,shape=NEbox,inner sep=2pt,minimum height=6mm,fill=white]
\tikzstyle{mapdag}=[draw,shape=SEbox,inner sep=2pt,minimum height=6mm,fill=white]
\tikzstyle{mapadj}=[draw,shape=SEbox,inner sep=2pt,minimum height=6mm,fill=white]
\tikzstyle{maptrans}=[draw,shape=SWbox,inner sep=2pt,minimum height=6mm,fill=white]
\tikzstyle{mapconj}=[draw,shape=NWbox,inner sep=2pt,minimum height=6mm,fill=white]

\tikzstyle{medium map}=[draw,shape=NEbox,inner sep=2pt,minimum height=6mm,fill=white,minimum width=7mm]
\tikzstyle{medium map dag}=[draw,shape=SEbox,inner sep=2pt,minimum height=6mm,fill=white,minimum width=7mm]
\tikzstyle{medium map adj}=[draw,shape=SEbox,inner sep=2pt,minimum height=6mm,fill=white,minimum width=7mm]
\tikzstyle{medium map trans}=[draw,shape=SWbox,inner sep=2pt,minimum height=6mm,fill=white,minimum width=7mm]
\tikzstyle{medium map conj}=[draw,shape=NWbox,inner sep=2pt,minimum height=6mm,fill=white,minimum width=7mm]
\tikzstyle{semilarge map}=[draw,shape=NEbox,inner sep=2pt,minimum height=6mm,fill=white,minimum width=9.5mm]
\tikzstyle{semilarge map trans}=[draw,shape=SWbox,inner sep=2pt,minimum height=6mm,fill=white,minimum width=9.5mm]
\tikzstyle{semilarge map adj}=[draw,shape=SEbox,inner sep=2pt,minimum height=6mm,fill=white,minimum width=9.5mm]
\tikzstyle{semilarge map dag}=[draw,shape=SEbox,inner sep=2pt,minimum height=6mm,fill=white,minimum width=9.5mm]
\tikzstyle{semilarge map conj}=[draw,shape=NWbox,inner sep=2pt,minimum height=6mm,fill=white,minimum width=9.5mm]
\tikzstyle{large map}=[draw,shape=NEbox,inner sep=2pt,minimum height=6mm,fill=white,minimum width=12mm]
\tikzstyle{very large map}=[draw,shape=NEbox,inner sep=2pt,minimum height=6mm,fill=white,minimum width=17mm]

\tikzstyle{medium dmap}=[draw,doubled,shape=NEbox,inner sep=2pt,minimum height=6mm,fill=white,minimum width=7mm]
\tikzstyle{medium dmap dag}=[draw,doubled,shape=SEbox,inner sep=2pt,minimum height=6mm,fill=white,minimum width=7mm]
\tikzstyle{medium dmap adj}=[draw,doubled,shape=SEbox,inner sep=2pt,minimum height=6mm,fill=white,minimum width=7mm]
\tikzstyle{medium dmap trans}=[draw,doubled,shape=SWbox,inner sep=2pt,minimum height=6mm,fill=white,minimum width=7mm]
\tikzstyle{medium dmap conj}=[draw,doubled,shape=NWbox,inner sep=2pt,minimum height=6mm,fill=white,minimum width=7mm]
\tikzstyle{semilarge dmap}=[draw,doubled,shape=NEbox,inner sep=2pt,minimum height=6mm,fill=white,minimum width=9.5mm]
\tikzstyle{semilarge dmap trans}=[draw,doubled,shape=SWbox,inner sep=2pt,minimum height=6mm,fill=white,minimum width=9.5mm]
\tikzstyle{semilarge dmap adj}=[draw,doubled,shape=SEbox,inner sep=2pt,minimum height=6mm,fill=white,minimum width=9.5mm]
\tikzstyle{semilarge dmap dag}=[draw,doubled,shape=SEbox,inner sep=2pt,minimum height=6mm,fill=white,minimum width=9.5mm]
\tikzstyle{semilarge dmap conj}=[draw,doubled,shape=NWbox,inner sep=2pt,minimum height=6mm,fill=white,minimum width=9.5mm]
\tikzstyle{large dmap}=[draw,doubled,shape=NEbox,inner sep=2pt,minimum height=6mm,fill=white,minimum width=12mm]
\tikzstyle{large dmap conj}=[draw,doubled,shape=NWbox,inner sep=2pt,minimum height=6mm,fill=white,minimum width=12mm]
\tikzstyle{large dmap trans}=[draw,doubled,shape=SWbox,inner sep=2pt,minimum height=6mm,fill=white,minimum width=12mm]
\tikzstyle{very large dmap}=[draw,doubled,shape=NEbox,inner sep=2pt,minimum height=6mm,fill=white,minimum width=19.5mm]

\tikzstyle{muxbox}=[draw,shape=rectangle,minimum height=3mm,minimum width=3mm,fill=white]
\tikzstyle{dmuxbox}=[muxbox,doubled]

\tikzstyle{box}=[draw,shape=rectangle,inner sep=2pt,minimum height=6mm,minimum width=6mm,fill=white]
\tikzstyle{dbox}=[draw,doubled,shape=rectangle,inner sep=2pt,minimum height=6mm,minimum width=6mm,fill=white]
\tikzstyle{dmap}=[draw,doubled,shape=NEbox,inner sep=2pt,minimum height=6mm,fill=white]
\tikzstyle{dmapdag}=[draw,doubled,shape=SEbox,inner sep=2pt,minimum height=6mm,fill=white]
\tikzstyle{dmapadj}=[draw,doubled,shape=SEbox,inner sep=2pt,minimum height=6mm,fill=white]
\tikzstyle{dmaptrans}=[draw,doubled,shape=SWbox,inner sep=2pt,minimum height=6mm,fill=white]
\tikzstyle{dmapconj}=[draw,doubled,shape=NWbox,inner sep=2pt,minimum height=6mm,fill=white]

\tikzstyle{ddmap}=[draw,doubled,dashed,shape=NEbox,inner sep=2pt,minimum height=6mm,fill=white]
\tikzstyle{ddmapdag}=[draw,doubled,dashed,shape=SEbox,inner sep=2pt,minimum height=6mm,fill=white]
\tikzstyle{ddmapadj}=[draw,doubled,dashed,shape=SEbox,inner sep=2pt,minimum height=6mm,fill=white]
\tikzstyle{ddmaptrans}=[draw,doubled,dashed,shape=SWbox,inner sep=2pt,minimum height=6mm,fill=white]
\tikzstyle{ddmapconj}=[draw,doubled,dashed,shape=NWbox,inner sep=2pt,minimum height=6mm,fill=white]

\boxshape{sNEbox}{0pt}{3pt}
\boxshape{sSEbox}{0pt}{-3pt}
\boxshape{sNWbox}{3pt}{0pt}
\boxshape{sSWbox}{-3pt}{0pt}
\tikzstyle{smap}=[draw,shape=sNEbox,fill=white]
\tikzstyle{smapdag}=[draw,shape=sSEbox,fill=white]
\tikzstyle{smapadj}=[draw,shape=sSEbox,fill=white]
\tikzstyle{smaptrans}=[draw,shape=sSWbox,fill=white]
\tikzstyle{smapconj}=[draw,shape=sNWbox,fill=white]

\tikzstyle{dsmap}=[draw,dashed,shape=sNEbox,fill=white]
\tikzstyle{dsmapdag}=[draw,dashed,shape=sSEbox,fill=white]
\tikzstyle{dsmaptrans}=[draw,dashed,shape=sSWbox,fill=white]
\tikzstyle{dsmapconj}=[draw,dashed,shape=sNWbox,fill=white]

\boxshape{mNEbox}{0pt}{10pt}
\boxshape{mSEbox}{0pt}{-10pt}
\boxshape{mNWbox}{10pt}{0pt}
\boxshape{mSWbox}{-10pt}{0pt}
\tikzstyle{mmap}=[draw,shape=mNEbox]
\tikzstyle{mmapdag}=[draw,shape=mSEbox]
\tikzstyle{mmaptrans}=[draw,shape=mSWbox]
\tikzstyle{mmapconj}=[draw,shape=mNWbox]

\tikzstyle{mmapgray}=[draw,fill=gray!40!white,shape=mNEbox]
\tikzstyle{smapgray}=[draw,fill=gray!40!white,shape=sNEbox]

\makeatletter
\pgfdeclareshape{cornerpoint}{
\inheritsavedanchors[from=rectangle] 
\inheritanchorborder[from=rectangle]
\inheritanchor[from=rectangle]{center}
\inheritanchor[from=rectangle]{north}
\inheritanchor[from=rectangle]{south}
\inheritanchor[from=rectangle]{west}
\inheritanchor[from=rectangle]{east}
\backgroundpath{
\southwest \pgf@xa=\pgf@x \pgf@ya=\pgf@y
\northeast \pgf@xb=\pgf@x \pgf@yb=\pgf@y

\pgfmathsetmacro{\pgf@shorten@left}{\pgfkeysvalueof{/tikz/shorten left}}
\pgfmathsetmacro{\pgf@shorten@right}{\pgfkeysvalueof{/tikz/shorten right}}

\pgfpathmoveto{\pgfpoint{0.5 * (\pgf@xa + \pgf@xb)}{\pgf@ya - 5pt}}
\pgfpathlineto{\pgfpoint{\pgf@xa - 8pt + \pgf@shorten@left}{\pgf@yb - 1.5 * \pgf@shorten@left}}
\pgfpathlineto{\pgfpoint{\pgf@xa - 8pt + \pgf@shorten@left}{\pgf@yb}}
\pgfpathlineto{\pgfpoint{\pgf@xb + 8pt - \pgf@shorten@right}{\pgf@yb}}
\pgfpathlineto{\pgfpoint{\pgf@xb + 8pt - \pgf@shorten@right}{\pgf@yb - 1.5 * \pgf@shorten@right}}
\pgfpathclose
}
}

\pgfdeclareshape{cornercopoint}{
\inheritsavedanchors[from=rectangle] 
\inheritanchorborder[from=rectangle]
\inheritanchor[from=rectangle]{center}
\inheritanchor[from=rectangle]{north}
\inheritanchor[from=rectangle]{south}
\inheritanchor[from=rectangle]{west}
\inheritanchor[from=rectangle]{east}
\backgroundpath{
\southwest \pgf@xa=\pgf@x \pgf@ya=\pgf@y
\northeast \pgf@xb=\pgf@x \pgf@yb=\pgf@y

\pgfmathsetmacro{\pgf@shorten@left}{\pgfkeysvalueof{/tikz/shorten left}}
\pgfmathsetmacro{\pgf@shorten@right}{\pgfkeysvalueof{/tikz/shorten right}}

\pgfpathmoveto{\pgfpoint{0.5 * (\pgf@xa + \pgf@xb)}{\pgf@yb + 5pt}}
\pgfpathlineto{\pgfpoint{\pgf@xa - 8pt + \pgf@shorten@left}{\pgf@ya + 1.5 * \pgf@shorten@left}}
\pgfpathlineto{\pgfpoint{\pgf@xa - 8pt + \pgf@shorten@left}{\pgf@ya}}
\pgfpathlineto{\pgfpoint{\pgf@xb + 8pt - \pgf@shorten@right}{\pgf@ya}}
\pgfpathlineto{\pgfpoint{\pgf@xb + 8pt - \pgf@shorten@right}{\pgf@ya + 1.5 * \pgf@shorten@right}}
\pgfpathclose
}
}

\makeatother

\pgfkeyssetvalue{/tikz/shorten left}{0pt}
\pgfkeyssetvalue{/tikz/shorten right}{0pt}

\tikzstyle{kpoint common}=[draw,fill=white,inner sep=1pt,minimum height=3mm]
\tikzstyle{kpoint}=[shape=cornerpoint,shorten left=5pt,kpoint common]
\tikzstyle{kpoint adjoint}=[shape=cornercopoint,shorten left=5pt,kpoint common]
\tikzstyle{kpoint conjugate}=[shape=cornerpoint,shorten right=5pt,kpoint common]
\tikzstyle{kpoint transpose}=[shape=cornercopoint,shorten right=5pt,kpoint common]
\tikzstyle{kpoint symm}=[shape=cornerpoint,shorten left=5pt,shorten right=5pt,kpoint common]

\tikzstyle{black kpoint}=[shape=cornerpoint,shorten left=5pt,kpoint common,fill=black]
\tikzstyle{black kpoint adjoint}=[shape=cornercopoint,shorten left=5pt,kpoint common,fill=black]

\tikzstyle{kpointdag}=[kpoint adjoint]
\tikzstyle{kpointadj}=[kpoint adjoint]
\tikzstyle{kpointconj}=[kpoint conjugate]
\tikzstyle{kpointtrans}=[kpoint transpose]

\tikzstyle{big kpoint}=[kpoint, minimum width=1.2 cm, minimum height=8mm, inner sep=4pt, text depth=3mm]

\tikzstyle{wide kpoint}=[kpoint, minimum width=1 cm, inner sep=2pt, text depth=-0.7 mm]
\tikzstyle{wide kpointdag}=[kpointdag, minimum width=1 cm, inner sep=2pt, text depth=0.7 mm]
\tikzstyle{wide kpointconj}=[kpointconj, minimum width=1 cm, inner sep=2pt, text depth=-0.7 mm]
\tikzstyle{wide kpointtrans}=[kpointtrans, minimum width=1 cm, inner sep=2pt, text depth=0.7 mm]

\tikzstyle{gray kpoint}=[kpoint,fill=gray!50!white]
\tikzstyle{gray kpointdag}=[kpointdag,fill=gray!50!white]
\tikzstyle{gray kpointadj}=[kpointadj,fill=gray!50!white]
\tikzstyle{gray kpointconj}=[kpointconj,fill=gray!50!white]
\tikzstyle{gray kpointtrans}=[kpointtrans,fill=gray!50!white]

\tikzstyle{gray dkpoint}=[kpoint,fill=gray!50!white,doubled]
\tikzstyle{gray dkpointdag}=[kpointdag,fill=gray!50!white,doubled]
\tikzstyle{gray dkpointadj}=[kpointadj,fill=gray!50!white,doubled]
\tikzstyle{gray dkpointconj}=[kpointconj,fill=gray!50!white,doubled]
\tikzstyle{gray dkpointtrans}=[kpointtrans,fill=gray!50!white,doubled]

\tikzstyle{white label}=[draw,fill=white,rectangle,inner sep=0.7 mm]
\tikzstyle{gray label}=[draw,fill=gray!50!white,rectangle,inner sep=0.7 mm]
\tikzstyle{black label}=[draw,fill=black,rectangle,inner sep=0.7 mm]

\tikzstyle{dkpoint}=[kpoint,doubled]
\tikzstyle{wide dkpoint}=[wide kpoint,doubled]
\tikzstyle{dkpointdag}=[kpoint adjoint,doubled]
\tikzstyle{dkcopoint}=[kpoint adjoint,doubled]
\tikzstyle{dkpointadj}=[kpoint adjoint,doubled]
\tikzstyle{dkpointconj}=[kpoint conjugate,doubled]
\tikzstyle{dkpointtrans}=[kpoint transpose,doubled]

\tikzstyle{kscalar}=[kpoint common, shape=EBox, inner xsep=-1pt, inner ysep=3pt,font=\small]
\tikzstyle{kscalarconj}=[kpoint common, shape=WBox, inner xsep=-1pt, inner ysep=3pt,font=\small]

 \tikzstyle{upground}=[circuit ee IEC,thick,ground,rotate=90,scale=2.5]
 \tikzstyle{downground}=[circuit ee IEC,thick,ground,rotate=-90,scale=2.5]
 \tikzstyle{bigground}=[regular polygon,regular polygon sides=3,draw=gray,scale=0.50,inner sep=-0.5pt,minimum width=10mm,fill=gray]

\tikzstyle{arrs}=[-latex,font=\small,auto]
\tikzstyle{arrow plain}=[arrs]
\tikzstyle{arrow dashed}=[dashed,arrs]
\tikzstyle{arrow bold}=[very thick,arrs]
\tikzstyle{arrow hide}=[draw=white!0,-]
\tikzstyle{arrow reverse}=[latex-]
\tikzstyle{cdnode}=[]

\begin{document}

	\title{Categorical Semantics for Schr\"{o}dinger's Equation}
	
	\author{
		\IEEEauthorblockN{Stefano Gogioso}
		\IEEEauthorblockA{
		Quantum Group, Dept of Computer Science\\
		University of Oxford, UK\\
		\href{mailto: stefano.gogioso@cs.ox.ac.uk}{stefano.gogioso@cs.ox.ac.uk}
		}
	}

	\maketitle

	\begin{abstract}
		Applying ideas from monadic dynamics to the well-established framework of categorical quantum mechanics, we provide a novel toolbox for the simulation of finite-dimensional quantum dynamics. We use strongly complementary structures to give a graphical characterisation of quantum clocks, their action on systems and the relevant energy observables, and we proceed to formalise the connection between unitary dynamics and projection-valued spectra. We identify the Weyl canonical commutation relations in the axioms of strong complementarity, and conclude the existence of a dual pair of time/energy observables for finite-dimensional quantum clocks, with the relevant uncertainty principle given by mutual unbias of the corresponding orthonormal bases. We show that Schr\"{o}dinger's equation can be abstractly formulated as characterising the Fourier transforms of certain Eilenberg-Moore morphisms from a quantum clock to a quantum dynamical system, and we use this to obtain a generalised version of the Feynman's clock construction. We tackle the issue of synchronism of clocks and systems, prove conservation of total energy and give conditions for the existence of an internal time observable for a quantum dynamical system. Finally, we identify our treatment as part of a more general theory of simulated symmetries of quantum systems (of which our clock actions are a special case) and their conservation laws (of which energy is a special case).
	\end{abstract}

	\setlength{\parindent}{0pt}
	\numberwithin{equation}{section}

\section{Background}

		The Categorical Quantum Mechanics (CQM) programme \cite{CQM-seminal}\cite{CQM-DeepBeauty}\cite{CQM-NewStructures}\cite{CQM-QCSnotes} is concerned with the understanding, through the language of category theory, of the structural and operational features of quantum theory. The sequential and parallel aspects of quantum processes are captured by the compositional and symmetric monoidal structures of the category $\fdHilbCategory$ of finite-dimensional Hilbert spaces and linear maps. State/operator duality, corresponding to the Shr\"{o}dinger/Heisemberg picture duality in quantum dynamics, finds its categorical formulation in the $\dagger$-compact structure of $\fdHilbCategory$.\\

		The formulation of Categorical Quantum Mechanics in terms of symmetric monoidal category theory makes it possible for it to be faithfully and rigorously translated into a graphical calculus \cite{QFT-BaezRosetta}. Based on a two-dimensional wire-and-box visualisation of processes, with wire composition giving information flow and juxtaposition used for parallelism, this pictorial formulation brings forth the operational aspects of quantum theory\cite{CQM-QuantumPicturalism}, leaving behind the clutter of traditional vector and matrix notation.\\

		The investigation of classical-quantum duality goes through the definition of classical structures \cite{CQM-QuantumClassicalStructuralism}\cite{CQM-QuantumMeasuNoSums}, i.e. special commutative $\dagger$-Frobenius algebras. Shown to correspond to orthonormal bases (which they delete, duplicate and match) in $\fdHilbCategory$, classical structures are key to the operational characterisation of quantum measurements, controlled operations and completely positive maps (environments and mixed state quantum mechanics are further explored in \cite{CQM-EvironmentClassicalChannels}). Interaction between different classical structures can be understood by introducing requirements of coherence, complementarity (a.k.a. Hopf law) and strong complementarity (a.k.a. bialgebra equations): strongly complementary structures play a fundamental role in the formulation of non-locality\cite{CQM-StrongComplementarity}, and are exactly classified by finite abelian groups \cite{CQM-KissingerPhdthesis}.\\

		The categorical approach to theories of physical systems and transformations is perhaps best argued in \cite{QFT-DoeringIsham4} and the collected works \cite{CQM-DeepBeauty}\cite{CQM-NewStructures}, and it provides the starting point for the formulation of the \textit{monadic dynamics} framework \cite{StefanoGogioso-MonadicDynamics}. In the same spirit of physical simulation as \cite{QTC-FeynmannSimulatingPhysicsComputers}\cite{QGR-QuantumGraphenity}, the framework aims at a categorical formulation of the operational aspects of dynamics, internalising notions of time via simulation by physical systems (a.k.a. clocks); dynamical systems are identified as the objects of the Eilenberg-Moore category for a certain monad encoding the dynamical structure of the clock. In the context of quantum theory, this idea yields a quantum clock similar to that presented in \cite{QTC-FeynmannSimulatingPhysicsComputers} (where it was called a \textit{time register}), with complex unitary representations of finite groups giving the quantum dynamical systems. Furthermore, monadic dynamics can be used to model quantum circuits, and yield the Feynman's clock construction as a corollary: the construction, originally introduced in \cite{QTC-FeynmannSimulatingPhysicsComputers} and recently revisited in \cite{QTC-FeynmannClockNewVariationalPrinciple}, plays a central role in simulated quantum dynamics, as it provides a way of computing whole histories of states in quantum circuits by finding the ground states of a certain Hamiltonian.

\section{Introduction}

		Categorical Quantum Mechanics aspires to be the formalism of choice for the description of quantum algorithms (see, for example, the work in \cite{CQM-TopologyQuantumAlgorithms}), but currently lacks adequate tools for the description of the dynamical aspects of quantum mechanics. The main aim of this paper is to cover this gap, opening the way for a systematic application of the framework to the field of simulated quantum dynamics. The technical details on representation theory in CQM are collected in the companion work \cite{StefanoGogioso-RepTheoryCQM}, which will be the main, unspoken reference for sections \ref{section_QuantumClocks} and \ref{section_EnergyObservable}.\\

		In section \ref{section_QuantumClocks}, we introduce quantum clocks in the CQM formalism, encode their time translation group structure in a strongly complementary pair of classical structures, and proceed to define quantum dynamical systems governed (or simulated) by them. In section \ref{section_EnergyObservable}, we identify the classical structures as observables for the time states and the energy levels of the clock, define Hamiltonians (seen as comonadic objects) as the adjoints of unitary dynamics (seen as monadic objects), provide the Fourier transform and prove that the eigenvalues of the Hamiltonians find their natural environment in the Pontryagin dual of the time translation group (their group of multiplicative characters). Furthermore, we give the proof of Von Neumann's mean ergodic theorem within the framework, identifying time-averages for the evolution of a quantum dynamical system with the ground states of the Hamiltonian.\\

		In section \ref{section_CCR}, we prove that strong complementarity is the same as canonical commutation, concluding that clocks necessarily possess a conjugate pair of time/energy observables. We define demolition measurements of states, and prove the time/energy uncertainty principle. In section \ref{section_SchrodingersEquation}, we give a categorical definition of Schr\"{o}dinger's equation as the energy perspective of the defining equation for Eilenberg-Moore morphisms. In section \ref{section_FeynmanClock}, we identify quantum circuits as certain composite quantum dynamical systems, define stationary flows of states through them and prove the validity of the Feynman's clock construction. In section \ref{section_Synchronisation}, we define synchronisation of clocks with quantum dynamical systems (the \inlineQuote{real world} counterpart for the quantum dynamics defined in this work), and show how \inlineQuote{forgetting} a clock induces a conserved total energy for the remaining synchronised systems. In section \ref{section_InternalTimeObservables}, we define internal time observables and give necessary and sufficient conditions for their existence. Finally, we give conditions under which a system in a synchronised family can be turned into a clock governing all the other systems.\\

		In section \ref{section_GeneratorsDynamics}, we draw the connection with Stone's Theorem on 1-parameter unitary groups, the result traditionally linking continuous-time dynamics to Hamiltonians and energy. In section \ref{section_GeneralDynamics} we briefly revisit this work in the light of a more general theory of symmetries of quantum systems, and explain how such a generalisation could be achieved in the framework of CQM. 

\newpage
\section{Quantum Clocks}
	\label{section_QuantumClocks}

	\begin{definition}\label{def_QuantumClock} 
		A \textbf{quantum clock} in $\fdHilbCategory$ is an internal monoid $(\;\timeobj\;, \;\raisebox{\SymFontShift}{\hbox{\input{modules/symbols/timemultSym.tex}}}, \;\raisebox{\SymFontShift}{\hbox{\input{modules/symbols/timeunitSym.tex}}})$ acting fully and faithfully on some orthonormal basis $(\ket{t})_{t:\integersMod{N}}$ of $\timeobj$, henceforth the \textbf{time basis}, as a cyclic group $\integersMod{N}$. We will refer to the binary operation $\raisebox{\SymFontShift}{\hbox{\input{modules/symbols/timemultSym.tex}}}$ as the \textbf{time translation}, to the state $\raisebox{\SymFontShift}{\hbox{\input{modules/symbols/timeunitSym.tex}}}$ as the \textbf{initial time} and we will often abuse notation and use $\timeobj$ for both the internal monoid and the underlying Hilbert space. 
	\end{definition}

	Concretely, when talking about a quantum clock we mean a quantum system with a specified time basis as in definition \ref{def_QuantumClock}. But in what sense does such a clock tick time for some dynamical system? In the monadic framework of \cite{StefanoGogioso-MonadicDynamics}, the answer is given by the associated dynamics.

	\begin{definition}\label{def_Dynamics}
		The \textbf{unitary dynamics} for a quantum clock $\timeobj$ are the actions $\hbox{\input{modules/symbols/algebraSym.tex}} \!\! : \SpaceH \tensor \timeobj \epim \SpaceH$ of the clock on quantum systems (i.e. the $\timeobj$-modules) defined by the following three equations:
		\begin{equation}\label{eqn_DynamicsDef1}
			\hbox{\input{modules/pictures/DynamicsDef1.tex}}
		\end{equation}
		\begin{equation}\label{eqn_DynamicsDef2}
			\hbox{\input{modules/pictures/DynamicsDef2.tex}}
		\end{equation}
		\begin{equation}\label{eqn_DynamicsDef3}
			\hbox{\input{modules/pictures/DynamicsDef3.tex}}
		\end{equation}
		where the $\hbox{\input{modules/symbols/antipodeSym.tex}}\!: \timeobj \rightarrow \timeobj$ is the \textbf{time inversion}, acting as group inverse on the time basis. In particular, $\raisebox{\SymFontShift}{\hbox{\input{modules/symbols/timemultSym.tex}}}$ is a unitary dynamic for the quantum clock itself, the \textbf{clock dynamic}.
	\end{definition}

	We will also refer to (unitary) $\timeobj$-modules as \textbf{quantum dynamical systems}, and note \cite{StefanoGogioso-MonadicDynamics} that they are exactly the objects of the Eilenberg-Moore category for the monad associated with the internal monoid $\timeobj$. They are, in particular, complex linear representations of the group $\integersMod{N}$ of time translations for the quantum clock.\\

	The internal structure of a quantum clock is categorically characterised by a strongly complementary pair of structures: a special commutative $\dagger$-Frobenius algebra identifying the time states and a commutative $\dagger$-Frobenius algebra giving the group structure. 

	\begin{theorem}\label{thm_TimeStructure}
		There is a canonical special commutative $\dagger$-Frobenius algebra (henceforth $\dagger$-SCFA) associated with the time basis, which we will call the \textbf{classical structure for time}, and unitary dynamics are controlled unitaries for it:
		\begin{equation}\label{eqn_TimeStructure}
			\hbox{\input{modules/pictures/TimeStructure.tex}}
		\end{equation}
	\end{theorem}
	\begin{proof}
		The relation between $\dagger$-SCFAs and orthonormal bases is proven in \cite{CQM-OrthogonalBases}. The controlled unitary axioms follow from the fact that $\raisebox{\SymFontShift}{\hbox{\input{modules/symbols/timemultSym.tex}}}$, $\hbox{\input{modules/symbols/antipodeSym.tex}}$ and $\raisebox{\SymFontShift}{\hbox{\input{modules/symbols/timediagSym.tex}}}$ satisfy Hopf law.
	\end{proof}

	\begin{theorem}\label{thm_GroupStructure}
		The quantum clock of def \ref{def_QuantumClock} can always be completed to a commutative $\dagger$-Frobenius algebra (henceforth $\dagger$-CFA) with antipode, which we will call the \textbf{group structure}: 
		\begin{equation}\label{eqn_GroupStructure}
			\hbox{\input{modules/pictures/GroupStructure.tex}}
		\end{equation}
		The antipode is formally defined to act by conjugation on the copiables of the $\dagger$-CFA (see thm \ref{thm_EnergyStructure} for confirmation of this). Furthermore, the group structure is strongly complementary dual to the classical structure for time, with antipode the time inversion.
	\end{theorem}
	\begin{proof}
		Section 7.2 of \cite{CQM-KissingerPhdthesis} contains the proof valid for all finite abelian groups, but the proof can be straightforwardly extended to finite non-abelian groups by dropping commutativity for the group structure.
	\end{proof}

	\begin{remark}\label{rmrk_GroupStructSpecial}
		The group structure is \textit{quasi} special, in the sense that composing comultiplication and multiplication gives the identity up to a normalisation factor of $N$. \footnote{Equivalently, all its copiables will have norm $N$. See \cite{CQM-OrthogonalBases}.} Most of the results and definitions involving classical structures will go through for the group structure, up to a normalisation factor of $N$ here and there, so we will often refer to it as a classical structure (term usually reserved to $\dagger$-SCFAs).
	\end{remark}

\section{The energy observable}
	\label{section_EnergyObservable}

	The classical structure for time characterises the time states as its copiables, but what about the group structure? It turns out that it defines something very interesting indeed.

	\begin{theorem}\label{thm_EnergyStructure}
		The \textbf{multiplicative characters} for the time translation group can be identified with the elements $\bra{\chi}$ of the dual space $\timeobj^\star$ satisfying the following equations:
		\begin{equation}\label{eqn_MultCharDef}
			\hbox{\input{modules/pictures/MultCharDef.tex}}
		\end{equation}
		The multiplicative characters form an orthogonal base for $\timeobj^\star$, and the group structure of thm \ref{thm_GroupStructure} is the canonical $\dagger$-CFA copying it. The comonoid $(\,\raisebox{\SymFontShift}{\hbox{\input{modules/symbols/timediagSym.tex}}} \!, \,\raisebox{\SymFontShift}{\hbox{\input{modules/symbols/trivialcharSym.tex}}}\!)$ acts on the multiplicative characters as their pointwise multiplication group, with the antipode as inverse (i.e. acting by conjugation).
	\end{theorem} 
	\begin{proof}
		All rather straightforward; more details and context can be found in \cite{StefanoGogioso-RepTheoryCQM}.
	\end{proof}

	We will refer to the multiplicative characters as the \textbf{energy levels (of the clock)}, because they take the familiar form
	\begin{equation}\label{eqn_EnergyLevels}
		\chi_E \eqdef \ket{t} \mapsto e^{i \, 2\pi / N \, E \cdot t}\text{ for } E : \integersMod{N}
	\end{equation}
	The trivial character $\chi_0 = \ket{t_j} \mapsto 1$ will be referred to as the \textbf{ground energy level}. Also, when understanding the classical structure for time as pointwise multiplication for the energy levels, we will refer to the group structure of thm \ref{thm_GroupStructure} as the \textbf{classical structure for energy}. In terms of associated bases, strong complementarity yields a duality between the time states and energy levels; in terms of group structures, it yields the duality between the time translation group $G = \integersMod{N}$ and its \textbf{Pontryagin dual} $G^\wedge \isom \integersMod{N}$, the pointwise multiplication group of the multiplicative characters of $G$.\\ 

	In CQM, non-degenerate observables are often identified with the $\dagger$-SCFAs canonically associated to their orthonormal bases of eigenstates; in order to accommodate the general, possibly degenerate case, one defines general observables as projector-valued spectra, dependent on a some $\dagger$-CFA for the labelling of projectors. 
	\begin{definition}\label{def_Spectra} 
		The \textbf{observables} for a $\dagger$-CFA with antipode $(\,\raisebox{\SymFontShift}{\hbox{\input{modules/symbols/timemultSym.tex}}}\!,\raisebox{\SymFontShift}{\hbox{\input{modules/symbols/timeunitSym.tex}}}\!,\raisebox{\SymFontShift}{\hbox{\input{modules/symbols/timecomultSym.tex}}}\!,\raisebox{\SymFontShift}{\hbox{\input{modules/symbols/timecounitSym.tex}}}\!,\hbox{\input{modules/symbols/antipodeSym.tex}}\!)$ are the maps $\hbox{\input{modules/symbols/measurementSym.tex}} \!\! : \SpaceH \rightarrow \SpaceH \tensor \timeobj$ which are self-adjoint, idempotent and complete:
		\begin{equation}\label{eqn_MeasurementsDef1}
			\hbox{\input{modules/pictures/MeasurementsDef1.tex}}
		\end{equation}
		\begin{equation}\label{eqn_MeasurementsDef2}
			\hbox{\input{modules/pictures/MeasurementsDef2.tex}}
		\end{equation}
		\begin{equation}\label{eqn_MeasurementsDef3}
			\hbox{\input{modules/pictures/MeasurementsDef3.tex}}
		\end{equation}
	\end{definition}
	
\newpage
	The \textbf{energy observables} are those associated with the classical structure for energy: the following two results use them to define Hamiltonians, and set the Pontryagin dual of the time translation group as the natural setting for the energy levels of a quantum dynamical system. Furthermore, in terms of Pontryagin duality, the Fourier transform is the statement that there is a canonical isomorphism $\Ltwo{G} \isom \Ltwo{G^\wedge}$: in our case the two $\LtwoSym$ spaces correspond to $\timeobj$ with the time basis and $\timeobj^\star$ with the basis of energy levels, yielding eqn \ref{eqn_MultCharResolutionCap} as the definition for the transform. 

	\begin{theorem}\label{thm_DynamicsSpectra}\textbf{(Hamiltonians)}\\
		The adjoints of unitary dynamics are exactly the energy observables. The energy observable associated with a unitary dynamic $\hbox{\input{modules/symbols/algebraSym.tex}}\!\!\!\!$ will be called its \textbf{Hamiltonian} and denoted by $\hbox{\input{modules/symbols/measurementSym.tex}}\!\!\!$. In particular, $\raisebox{\SymFontShift}{\hbox{\input{modules/symbols/timecomultSym.tex}}}\!\!\!$ is the Hamiltonian for the clock dynamic $\raisebox{\SymFontShift}{\hbox{\input{modules/symbols/timemultSym.tex}}}\!\!$.
	\end{theorem}
	\begin{proof}
		Immediate by taking adjoints. More details and context can be found in \cite{StefanoGogioso-RepTheoryCQM}.
	\end{proof}

	\begin{theorem}\label{thm_FourierTransform}\textbf{(Fourier transform)}\\
		The energy levels form a partition of the counit, i.e. satisfy $\frac{1}{N} \sum_\chi \!\!\!\!\text{
		\renewcommand{\repElementSymName}{$\chi$}
		\raisebox{\SymFontShift}{\hbox{\input{modules/symbols/multCharacterEffectSym.tex}}}
		}\!\! = \raisebox{\SymFontShift}{\hbox{\input{modules/symbols/timecounitSym.tex}}}$. As a consequence, any quantum dynamical system endowed with a unitary dynamic is covered, via the Hamiltonian, by a complete family of projectors, labelled by the energy levels: 
		\begin{equation}\label{eqn_EnergySpectraResolutionId}
			\hbox{\input{modules/pictures/EnergySpectraResolutionId.tex}}
		\end{equation}
		The energy levels corresponding to non-zero-dimensional projectors will be called the \textbf{energy levels of the system}. In particular, the Hamiltonian for the clock dynamic gives a resolution of the identity in terms of the energy levels:
		\begin{equation}\label{eqn_MultCharResolutionId}
			\hbox{\input{modules/pictures/MultCharResolutionId.tex}}
		\end{equation}
		The transpose of the resolution in eqn \ref{eqn_MultCharResolutionId} gives the \textbf{Fourier transform} for the clock:
		\begin{equation}\label{eqn_MultCharResolutionCap}
			\hbox{\input{modules/pictures/MultCharResolutionCap.tex}}
		\end{equation}
	\end{theorem}
	\begin{proof}
		The details can be found in \cite{StefanoGogioso-RepTheoryCQM}.
	\end{proof}

	\begin{definition}\label{def_EigenstatesEigenvalues}
		The result of thm \ref{thm_FourierTransform} holds more in general for any observable, with copiables of the classical structure labelling the projectors. We can define the \textbf{eigenstates} for the observable to be the eigenstates for the projectors, and the \textbf{eigenvalue} $\ket{\chi}$ for an eigenstate $\ket{\psi}$ to be the copiable labelling the relevant projector:
		\begin{equation}\label{eqn_EigenstatesEigenvalues}
			\hbox{\input{modules/pictures/EigenstatesEigenvaluesDef.tex}}
		\end{equation}
		In particular one obtains the energy eigenstates, with the energy levels as eigenvalues. If the classical structure is a $\dagger$-SCFA, then the normalisation factor of $1/N$ in the partition of the counit should be omitted.
	\end{definition}

	Finally we are in the position to formulate and prove the ergodic theorem for quantum dynamics.
	\begin{theorem}\label{thm_VonNeumannErgodicTheorem}\textbf{(Von Neumann's mean ergodic theorem)}\\
		The time-average for the evolution of a quantum dynamical system is the projector over the ground states:\footnote{Henceforth the eigenstates associated with the ground energy level will be called \textbf{ground states}.}
		\begin{equation}\label{eqn_ErgodicTheorem}
			\hbox{\input{modules/pictures/ErgodicTheorem.tex}}
		\end{equation}
	\end{theorem}
	\begin{proof}
		\begin{equation}\label{eqn_ErgodicTheoremProof}
			\hbox{\input{modules/pictures/ErgodicTheoremProof.tex}}
		\end{equation}
	\end{proof}

\section{Canonical commutation}		
	\label{section_CCR}

	First, let's extend strong complementarity to observables.
	\begin{definition}\label{def_WeylCCRs}
		Let $\hbox{\input{modules/symbols/measurementSymGroupStructColour.tex}}\!\!\!\!$ and $\hbox{\input{modules/symbols/measurementSymClassicalStructColour.tex}}\!\!\!\!$ be observables, on the same quantum system $\SpaceH$, for a pair of strongly complementary classical structures $(\,\raisebox{\SymFontShift}{\hbox{\input{modules/symbols/timemultSym.tex}}}\!,\raisebox{\SymFontShift}{\hbox{\input{modules/symbols/timeunitSym.tex}}}\!,\raisebox{\SymFontShift}{\hbox{\input{modules/symbols/timecomultSym.tex}}}\!,\raisebox{\SymFontShift}{\hbox{\input{modules/symbols/timecounitSym.tex}}}\!)$ and $(\,\raisebox{\SymFontShift}{\hbox{\input{modules/symbols/timematchSym.tex}}}\!,\raisebox{\SymFontShift}{\hbox{\input{modules/symbols/timematchunitSym.tex}}}\!,\raisebox{\SymFontShift}{\hbox{\input{modules/symbols/timediagSym.tex}}}\!,\raisebox{\SymFontShift}{\hbox{\input{modules/symbols/trivialcharSym.tex}}}\!)$, with antipode $\hbox{\input{modules/symbols/antipodeSym.tex}}\!$. We say that $\hbox{\input{modules/symbols/measurementSymGroupStructColour.tex}}\!\!\!\!$ and $\hbox{\input{modules/symbols/measurementSymClassicalStructColour.tex}}\!\!\!\!$ are \textbf{strongly complementary} if they satisfy the following bialgebra equation:
		\begin{equation}\label{eqn_CCRDef}
			\hbox{\input{modules/pictures/CCRDef.tex}}
		\end{equation}
		In particular, $\raisebox{\SymFontShift}{\hbox{\input{modules/symbols/timecomultSym.tex}}}\!\!$ and $\raisebox{\SymFontShift}{\hbox{\input{modules/symbols/timediagSym.tex}}}\!\!$ are strongly complementary.
	\end{definition}

\newpage
	Canonically commuting observables play a special role in quantum mechanics, as they relate symmetry groups with their conserved quantities. The standard definition involves elements of the Lie algebra, but there is an equivalent definition in terms of the 1-parameter unitary groups (unitary dynamics, in our case), the \textbf{Weyl canonical commutation relations}. The following theorem shows that canonical commutation, in the sense of the Weyl relations, is equivalent to strong complementarity.

	\begin{theorem}\label{thm_WeylCCRs}\textbf{(Weyl canonical commutation relation)}\\
		Two observables $\hbox{\input{modules/symbols/measurementSymGroupStructColour.tex}}\!\!\!\!$ and $\hbox{\input{modules/symbols/measurementSymClassicalStructColour.tex}}\!\!\!\!$ are strongly complementary if and only if the associated unitary dynamics are \textbf{canonically commuting}, i.e. satisfy the following equation:
		\begin{equation}\label{eqn_WeylCCRDef}
			\hbox{\input{modules/pictures/WeylCCRDef.tex}}
		\end{equation}
		for any $\ket{\chi}$ and $\ket{t}$ eigenvalues for $\hbox{\input{modules/symbols/measurementSymGroupStructColour.tex}}\!\!\!\!$ and $\hbox{\input{modules/symbols/measurementSymClassicalStructColour.tex}}\!\!\!\!$.
	\end{theorem}
	\begin{proof}
		\begin{equation}\label{eqn_WeylCCRProof}
			\hbox{\input{modules/pictures/WeylCCRProof.tex}}
		\end{equation}
		with application of eqn \ref{eqn_MeasurementsDef1} at the beginning and end.
	\end{proof}

	Thus a clock always has a canonically commuting pair of time/energy observables. Also, we would expect an associated uncertainty principle: to formulate it, we first define demolition measurements of individual states with respect to a given observable.\footnote{The full definition of demolition measurements requires the CPM construction, and can be found in \cite{CQM-QuantumMeasuNoSums} \cite{CQM-QCSnotes}. Here we only need the definition for a given state, so the CPM construction is not necessary.}

	\begin{definition}\label{def_ClassicalMeasurement}
		Given an observable $\hbox{\input{modules/symbols/measurementSym.tex}}\!\!\!\!$, the associated \textbf{demolition measurement} of state $\ket{\psi}$ is the map:
		\begin{equation}\label{eqn_DemolitionMeasurementsDef}
			\hbox{\input{modules/pictures/DemolitionMeasurementsDef.tex}}
		\end{equation}
	\end{definition}

	From eqn \ref{eqn_MeasurementsDef3}, it follows that the demolition measurement gives a distribution on the copiables of the classical structure (exactly a distribution for $\dagger$-SCFAs, or up to a normalisation factor of $1/N$ in the case of the group structure). It is in terms of this distribution that we prove the uncertainty principle for strongly complementary observables: any eigenstate of one of the two spectra is completely unbiased when measured with respect to the other spectrum.

	\begin{theorem}\label{thm_UncertaintyPrinciple}\textbf{(Uncertainty principle)}\\
		Let $\hbox{\input{modules/symbols/measurementSymGroupStructColour.tex}}\!\!\!\!$ and $\hbox{\input{modules/symbols/measurementSymClassicalStructColour.tex}}\!\!\!\!$ be canonically commuting observables. Then the demolition measurement with respect to $\hbox{\input{modules/symbols/measurementSymClassicalStructColour.tex}}\!\!\!\!$ of a normalised eigenstate $\ket{\psi}$ of $\hbox{\input{modules/symbols/measurementSymGroupStructColour.tex}}\!\!\!\!$ gives the uniform distribution. 
	\end{theorem}
	\begin{proof}
		We prove the statement with respect to an arbitrary eigenvalue $\ket{t}$ of $\hbox{\input{modules/symbols/measurementSymClassicalStructColour.tex}}\!\!\!\!$, assuming it to be self-conjugate without loss of generality. The factor of $N$ on the LHS of eqn \ref{eqn_UncertaintyPrincipleProof1} takes care of the fact that the group structure we are interested in is special only up to a factor of $N$; in the case where both structures are special, one omits the factor.
		\begin{equation}\label{eqn_UncertaintyPrincipleProof1}
			\hbox{\input{modules/pictures/UncertaintyPrincipleProof1.tex}}
		\end{equation}
		\begin{equation}\label{eqn_UncertaintyPrincipleProof2}
			\hbox{\input{modules/pictures/UncertaintyPrincipleProof2.tex}}
		\end{equation}
		\begin{equation}\label{eqn_UncertaintyPrincipleProof3}
			\hbox{\input{modules/pictures/UncertaintyPrincipleProof3.tex}}
		\end{equation}
		The RHS of eqn \ref{eqn_UncertaintyPrincipleProof3} reduces to the identity in the case of our group structure and the classical structure for time (using eqn \ref{eqn_DynamicsDef2} and the fact that multiplicative characters send time states to phases). If both structures are special, then the two conjugate scalars multiply to $1/N$ by mutual unbias of strongly complementary orthonormal bases \cite{CQM-ZXCalculusSeminal}, which compensates the absence of the factor of $N$ on the LHS of eqn \ref{eqn_UncertaintyPrincipleProof1}.
	\end{proof}


\section{Schr\"{o}dinger's Equation}
	\label{section_SchrodingersEquation}

	The trajectories of states in a quantum dynamical system $\hbox{\input{modules/symbols/algebraSym.tex}}\!\!\!\!$ have a nice categorical characterisation as the morphisms of dynamical systems from the clock to $\hbox{\input{modules/symbols/algebraSym.tex}}\!\!\!\!$.

	\begin{definition}\label{def_ConcreteHistories}
		The \textbf{concrete histories} (of states) for a quantum dynamical system are the trajectories of the states under the dynamics:
		\begin{equation}\label{eqn_ConcreteHistoriesEMDef}
			\hbox{\input{modules/pictures/ConcreteHistoriesEMDef.tex}}
		\end{equation}
	\end{definition}
	
	\begin{theorem}\label{thm_ConcreteHistories}
		The concrete histories for a quantum dynamical system are exactly the Eilenberg-Moore morphisms (i.e. morphisms of $\timeobj$-modules) $\raisebox{\SymFontShift}{\hbox{\input{modules/symbols/timemultSym.tex}}}\!\! \rightarrow \hbox{\input{modules/symbols/algebraSym.tex}}$:
		\begin{equation}\label{eqn_EMmorphismDef}
			\hbox{\input{modules/pictures/EMmorphismDef.tex}}
		\end{equation}
	\end{theorem}
	\begin{proof}
		In one direction, every concrete history $\Psi$ of some state $\ket{\psi}$ is an Eilenberg-Moore morphism $\raisebox{\SymFontShift}{\hbox{\input{modules/symbols/timemultSym.tex}}}\!\! \rightarrow \hbox{\input{modules/symbols/algebraSym.tex}}$:
		\begin{equation}\label{eqn_ConcreteHistoriesEMProof1}
			\hbox{\input{modules/pictures/ConcreteHistoriesEMProof1.tex}}
		\end{equation}
		In the other direction, every Eilemberg-Moore morphism $\Psi: \raisebox{\SymFontShift}{\hbox{\input{modules/symbols/timemultSym.tex}}}\!\! \rightarrow \hbox{\input{modules/symbols/algebraSym.tex}}$ is the concrete history of some state $\ket{\psi}$:
		\begin{equation}\label{eqn_ConcreteHistoriesEMProof2}
			\hbox{\input{modules/pictures/ConcreteHistoriesEMProof2.tex}}
		\end{equation}
	\end{proof}

	For concrete histories, seen as trajectories $\Psi: \timeobj \rightarrow \SpaceH$ of states under the dynamics, the focus is on the time basis:
	\begin{equation}
		\Psi = \ket{t} \rightarrow \ket{\psi_t}
	\end{equation}
	On the other hand, shifting the perspective to the energy basis (i.e. taking the Fourier transform) gives us the fundamental equation of quantum dynamics. 

	\begin{theorem}\label{thm_SchrodingersEqn}\textbf{(Schr\"{o}dinger's Equation)}\\
		The (exponentiated version of) Schr\"{o}dinger's equation for a given Hamiltonian $\hbox{\input{modules/symbols/measurementSym.tex}}\!\!\!\!$ can be written, in terms of the energy basis, as the defining equation of Eilenberg-Moore morphisms $\raisebox{\SymFontShift}{\hbox{\input{modules/symbols/timemultSym.tex}}}\!\! \rightarrow \hbox{\input{modules/symbols/algebraSym.tex}}$. The \textbf{solutions of Schr\"{o}dinger's equation} are therefore exactly the concrete histories, seen as functions of the energy levels.
	\end{theorem}
	\begin{proof}
		The exponentiated formulation of Schr\"{o}dinger's equation for the discrete periodic case is:
		\begin{equation}\label{eqn_SchrodingersEqnTraditional}
			U(t)\ket{\psi_E} = e^{i \, 2\pi / N \, E \cdot t} \; \ket{\psi_E}
		\end{equation}
		where $\psi_E$ is a function of energy levels. The formulation above corresponds to the first equality in the following graphical equation:
		\begin{equation}\label{eqn_SchrodingersEqnProof}
			\hbox{\input{modules/pictures/SchrodingersEqnProof.tex}}
		\end{equation}
		As the leftmost and rightmost maps in eqn \ref{eqn_SchrodingersEqnProof} coincide when applied to an arbitrary $\ket{\chi}$ and $\ket{t}$, we conclude that solutions $\Psi$ to Schr\"{o}dinger's equation are exactly the Eilenberg-Moore morphisms $\Psi: \raisebox{\SymFontShift}{\hbox{\input{modules/symbols/timemultSym.tex}}}\!\! \rightarrow \hbox{\input{modules/symbols/algebraSym.tex}}$, seen from the perspective of the energy levels.
	\end{proof}

	\begin{corollary}\label{thm_SchrodingersEqnComponents}
		A solution of Schr\"{o}dinger's equation evaluated at an energy level $\ket{\chi}$ gives a (possibly zero) eigenstate for the Hamiltonian corresponding to that energy level. We shall refer to this eigenstate as the $\ket{\chi}$-\textbf{spectral component} of the solution.
	\end{corollary}

	\begin{remark}\label{rmrk_SchrodingersEqn}
		The solutions to Schr\"{o}dinger's equation as defined in thm \ref{thm_SchrodingersEqn} are exactly the spectral decompositions of the trajectories of states in terms of energy levels. This contrasts with the the more traditional decomposition in terms of a basis of eigenstates of the Hamiltonian: in the case of degenerate Hamiltonians this may look like a bug, but it isn't. The decomposition in terms of a basis of eigenstates is not canonical, and the categorically correct object to consider is the decomposition in terms of the energy eigenspace projectors given by eqn \ref{eqn_EnergySpectraResolutionId}, which is exactly what the solutions to Schr\"{o}dinger's equation, as we defined them, provide.
	\end{remark}

\section{Feynman's Clock}
	\label{section_FeynmanClock}

	Given a quantum circuit composed of unitary gates, the Feynman's clock construction provides a Hamiltonian with ground states characterising the circuit dynamics. More precisely, if $(U_i)_{i=1,...,n-1}$ is some finite sequence of unitary gates on a quantum system $\SpaceH$, then the Hamiltonian has the following ground states:
	\begin{equation}\label{eqn_HistoryStatesEqn}
		\left[\sum\limits_{i=0,...,n} \ket{\psi_i} \tensor \ket{i} \right] \text{ s.t. } U_{i+1} \ket{\psi_i} = \ket{\psi_{i+1}}
	\end{equation}

	We consider the slightly different case of cyclical circuits $(U_t)_{t:\integersMod{N}}$, which can be seen as unitary dynamics for the composite quantum system $\SpaceH \tensor \timeobj$:
	\begin{equation}
		\ket{\psi} \tensor \ket{t} \tensor \ket{\Delta t} \mapsto \left[\left(\prod\limits_{s=t+1}^{t+\Delta t}U_{s}\right)\ket{\psi} \right] \tensor \ket{t+\Delta t}
	\end{equation}
	where the product is expanded to the left.
	\begin{definition}\label{def_Propagator}
		A \textbf{cyclic quantum circuit} is a unitary dynamic $\hbox{\input{modules/symbols/algebraSym.tex}}\!\!\!\!: \SpaceH \tensor \timeobj \tensor \timeobj \epim \SpaceH \tensor \timeobj$ taking the following form when applied to the generator $\ket{1}$ of $\integersMod{N}$:
		\begin{equation}\label{eqn_CyclicQuantumCircuitsDef}
			\hbox{\input{modules/pictures/CyclicQuantumCircuitsDef.tex}}
		\end{equation}
		The $U$ map on the RHS is a controlled unitary (over the classical structure for time) which encodes the gates of the circuit, while the clock $\timeobj$ indexes both the gates and the stages of the circuit.
	\end{definition}

	\begin{remark} The original construction can be obtained by extending the circuit to a cyclical one (e.g. concatenating it with its adjoint, taking $N=2n$) and then considering the (now intermediate) state $\ket{\psi_n}$ as before. 
	\end{remark}

	\begin{definition}\label{def_CircuitFlows}
		The \textbf{stationary flows} of a cyclic quantum circuit $\hbox{\input{modules/symbols/algebraSym.tex}}\!\!\!\!$ are the solutions $\Psi$ to Schr\"{o}dinger's equation taking the following form:
		\begin{equation}\label{eqn_CyclicQuantumCircuitsDynamics}
			\hbox{\input{modules/pictures/CyclicQuantumCircuitsDynamics.tex}}
		\end{equation}
		Then associated \textbf{history states}, i.e. the states satisfying eqn \ref{eqn_HistoryStatesEqn} and used in the original formulation of Feynman's clock, are defined to be the ground energy components:
		\begin{equation}\label{eqn_CyclicQuantumCircuitsFixedPoints}
			\hbox{\input{modules/pictures/CyclicQuantumCircuitsFixedPoints.tex}}
		\end{equation}
	\end{definition}
	
	We are now in the position of proving the fundamental result for the construction: in order to simulate the dynamics of a cyclic quantum circuit (seen as a unitary dynamic on $\SpaceH \tensor \timeobj$), all one has to do is find the ground states for its Hamiltonian.
	\begin{theorem}\label{thm_FeynmanClock}\textbf{(Feynman's Clock)}\\
		Given a cyclic quantum circuit $\hbox{\input{modules/symbols/algebraSym.tex}}\!\!$, the history states associated with the stationary flows are exactly the ground states of the Hamiltonian $\hbox{\input{modules/symbols/measurementSym.tex}}\!\!\!\!$. 
	\end{theorem}
	\begin{proof}
		In one direction: a history state is the ground energy spectral component of a solution to Schr\"{o}dinger's equation, and thus by corollary \ref{thm_SchrodingersEqnComponents} it is a ground state for the Hamiltonian. In the other direction: given a ground state $\ket{\psi}$ for the Hamiltonian, we will construct a $\Psi$ map in the form of eqn \ref{eqn_CyclicQuantumCircuitsDynamics} (using state-operator duality) and prove that it satisfies eqn \ref{eqn_ConcreteHistoriesEMDef}, i.e. that it is a solution to Schr\"{o}dinger's equation. As the time translation group for the clock is cyclic, it suffices to prove the result for the generator $\ket{1}$ (and then apply eqn \ref{eqn_DynamicsDef1}):
		\begin{equation}\label{eqn_FeynmanClockProof1}
			\hbox{\input{modules/pictures/FeynmanClockProof1.tex}}
		\end{equation}
		\begin{equation}\label{eqn_FeynmanClockProof2}
			\hbox{\input{modules/pictures/FeynmanClockProof2.tex}}
		\end{equation}
		The first and penultimate step use the definition of cyclic quantum circuit, while the last step uses the fact that $\ket{\psi}$ is a ground state for the Hamiltonian, and thus a fixed point for the unitary dynamic. The step between the two eqns is given by strong complementarity.
	\end{proof}

\section{Clock-system Synchronisation}
	\label{section_Synchronisation}

	Up to this point, we have described the action of a clock on physical systems in terms of unitary dynamics: although appealing from an algebraic point of view, this formulation lacks an immediate physical interpretation. To fix this, we shift point of view from the \textit{action} of a clock system on another physical system to the \textbf{synchronicity} of the two systems: instead of saying that a time state $\ket{t}$ of the clock is mapped by a concrete history to a state $\ket{\psi_t}$, we will say that whenever the clock is measured to be in time state $\ket{t}$, the system \inlineQuote{collapses} into state $\ket{\psi_t}$. This perspective leads us to the following definition.

	\begin{definition}\label{def_SynchronisedClockSystemPair}
		Let $\hbox{\input{modules/symbols/algebraSym.tex}}\!\!\!\!: \SpaceH \tensor \timeobj \epim \SpaceH$ be a quantum dynamical system, governed by clock $\timeobj$. Then the corresponding \textbf{synchronised clock-system pair} with intial state $\ket{\psi}$ is the following state of $\SpaceH \tensor \timeobj$:
		\begin{equation}\label{eqn_SyncClockSystemPair}
			\hbox{\input{modules/pictures/SyncClockSystemPair.tex}}
		\end{equation}
	\end{definition}

	When talking about the evolution of a quantum dynamical system, seen as a synchronised clock-system pair, what is usually meant by the energy observable is the Hamiltonian for the system, and what is usually meant by the time observable is the time observable for the clock. This simple observation, coupled with the following result, answers a 1930s conundrum by Schr\"{o}dinger (see \cite{Msc-TimeQTProspectus} p6).

	\begin{theorem}\label{thm_SchrodingersConundrum}
		In a synchronised clock-system pair the time observable for the clock and the Hamiltonian for the system commute. As a consequence, it is possible to simultaneously measure the time in the evolution of a system and its energy.
	\end{theorem}

	We now move on to model multiple systems synchronised with the same clock, by defining a notion of separability for dynamical systems.

	\begin{definition}\label{def_SeparableDynamic}
		We say that a concrete history for a quantum dynamical system $\hbox{\input{modules/symbols/algebraSym.tex}}\!\!\!\!: \SpaceH \tensor \timeobj \epim \SpaceH$ is \textbf{separable over the decomposition $\SpaceH = \SpaceH_1 \tensor ... \tensor \SpaceH_M$} if it takes the form:
		\begin{equation}\label{eqn_SeparableDynamic}
			\hbox{\input{modules/pictures/SeparableDynamic.tex}}
		\end{equation}
		where $(\alpha_j:\SpaceH_j \tensor \timeobj \epim \SpaceH_j)_{j=1,...,M}$ is a family of quantum dynamical systems all governed by clock $\timeobj$.
	\end{definition}

	Def \ref{def_SynchronisedClockSystemPair}, applied to an appropriately separable dynamical system, gives the desired synchronisation of multiple systems with a clock. But what happens if we forget about the clock, for example by fixing its energy state? We get a synchronised family of dynamical systems, with a well defined total energy. 

	\begin{definition}\label{def_SynchronisedFamily}
		Let $(\alpha_j:\SpaceH_j \tensor \timeobj \epim \SpaceH_j)_{j=1,...,M+1}$ be a family of quantum dynamical systems governed by some common clock $\timeobj$. Then the corresponding \textbf{synchronised family $(\alpha_j)_j$ with initial states $\ket{\psi_j}_j$ and total energy $\ket{\chi}$} is the following state in $\SpaceH_1 \tensor ... \tensor \SpaceH_{M+1}$: 
		\begin{equation}\label{eqn_SyncDynamicalSystems}
			\hbox{\input{modules/pictures/SyncDynamicalSystems.tex}}
		\end{equation}
		If $\alpha_{M+1} = \raisebox{\SymFontShift}{\hbox{\input{modules/symbols/timemultSym.tex}}}$, $\ket{\psi_{M+1}} = \raisebox{\SymFontShift}{\hbox{\input{modules/symbols/timeunitSym.tex}}}$ and $\bra{\chi} = \raisebox{\SymFontShift}{\hbox{\input{modules/symbols/trivialcharSym.tex}}}\!$, then this is equivalent to a synchronised system-clock pair where both the system $\hbox{\input{modules/symbols/algebraSym.tex}}\!\!\!\!: \SpaceH \tensor \timeobj \epim \SpaceH$ and the initial state $\ket{\psi}$ are separable over the decomposition $\SpaceH = \SpaceH_1 \tensor ...\tensor \SpaceH_M$.
	\end{definition}

	\begin{theorem}\label{thm_ConservationEnergy1}
		Consider a synchronised system-clock pair where both the system $\hbox{\input{modules/symbols/algebraSym.tex}}\!\!\!\!: \SpaceH \tensor \timeobj \epim \SpaceH$ and the initial state $\ket{\psi}$ are separable over the decomposition $\SpaceH = \SpaceH_1 \tensor ...\tensor \SpaceH_M$. Measuring the clock energy level to be $\ket{-\chi}$ yields a synchronised family $(\alpha_j)_{j=1,...,M}$ with initial states $\ket{\psi_j}_j$ and total energy $\ket{\chi}$. Concretely, as a state of $\SpaceH = \SpaceH_1 \tensor ...\tensor \SpaceH_M$, this results in the following superposition of energy eigenstates:
		\begin{equation}\label{eqn_ConservationEnergy1}
			\hbox{\input{modules/pictures/ConservationEnergy1.tex}}
		\end{equation}
	\end{theorem}
	\begin{proof}
		Consider a synchronised system-clock pair like that of eqn \ref{eqn_SyncClockSystemPair}, with the system in the form of eqn \ref{eqn_SeparableDynamic}. Then measuring the clock in the energy level $\ket{-\chi}$ results in the situation of eqn \ref{eqn_SyncDynamicalSystems}, and one obtains the superposition in eqn \ref{eqn_ConservationEnergy1} by observing that $\raisebox{\SymFontShift}{\hbox{\input{modules/symbols/timediagSym.tex}}}$ acts as a comultiplication for energy levels.
	\end{proof}

	Furthermore, the total energy is conserved, in the sense that measuring one of the systems to be in a specific energy eigenstates simply subtracts from the total energy of the remaining systems.

	\begin{theorem}\label{thm_ConservationEnergy2}\textbf{(Conservation of total energy)}\\
		Given a synchronised family $(\alpha_j)_{J=1,...,M+1}$ with initial states $\ket{\psi_j}_j$ and total energy $\ket{\chi}$, measuring the system $\alpha_{M+1}$ in an energy eigenstate, not orthogonal to $\psi_{M+1}$ and corresponding to energy level $\ket{\chi'}$, yields (possibly up to scaling/phase) a synchronised family $(\alpha_j)_{j=1,...,M}$ with initial states $\ket{\psi_j}_j$ and total energy $\ket{\chi-\chi'}$.
	\end{theorem}
	\begin{proof}
		Similar to that of thm \ref{thm_ConservationEnergy1}, using eqn \ref{eqn_EigenstatesEigenvalues} to turn the energy eigenstate into the corresponding energy level.
	\end{proof}

\section{Internal time observables}
	\label{section_InternalTimeObservables}

	We have seen that synchronisation allows us to give a physical description of the action of a clock on quantum dynamical systems, and that forgetting the clock imposes a conserved total energy on the remaining synchronised systems. One question remains open: under which conditions can one of these remaining synchronised systems turn itself into an \inlineQuote{internal} clock governing the other systems? This is an important question, for reasons explained in the end by remark \ref{rmrk_ContinuousSystems}, and this section is devoted to answering it.\\

	To begin with, a clock has a strongly complementary pair of time/energy observables, so we need to define these for our candidate internal clock. 
	\begin{definition}\label{def_InternalTimeObservable}
		Given a quantum dynamical system $\hbox{\input{modules/symbols/internalalgebraSym.tex}}\!\!\!\!: \SpaceH \tensor \timeobj \epim \SpaceH$, an \textbf{internal time observable} is a $\dagger$-SCFA $\internaltimestructure$ on $\SpaceH$ s.t. $\hbox{\input{modules/symbols/internalalgebraSym.tex}}\!\!\!\!$ is a classical map $(\internaltimecomonoid) \times (\raisebox{\SymFontShift}{\hbox{\input{modules/symbols/timediagSym.tex}}}\!,\raisebox{\SymFontShift}{\hbox{\input{modules/symbols/trivialcharSym.tex}}}\!) \epim (\internaltimecomonoid)$, necessarily corresponding to a classical group action of the time translation group on the corresponding orthonormal basis.
	\end{definition}

	\begin{definition}\label{def_NondegenerateSystem}
		We say that a Hamiltonian $\hbox{\input{modules/symbols/internalmeasurementSym.tex}}\!\!\!\!: \SpaceH \rightarrow \SpaceH \tensor \timeobj$ is \textbf{non-degenerate} if there is a unique $\dagger$-SCFA $\internalenergystructure$, the \textbf{(normalised) internal energy observable}, making $\hbox{\input{modules/symbols/internalmeasurementSym.tex}}\!\!\!\!$ into a classical map $\internalenergycomonoid \rightarrow \internalenergycomonoid \times (\raisebox{\SymFontShift}{\hbox{\input{modules/symbols/timecomultSym.tex}}}\!,\raisebox{\SymFontShift}{\hbox{\input{modules/symbols/timecounitSym.tex}}}\!)$. As it fits better with our original framework, we will also use $\internalenergystructure$ to denote the \textbf{internal energy observable}, having the same copiables as the normalised internal energy structure, but all scaled by a factor of $\Dim{\SpaceH}$.
	\end{definition}

	\begin{theorem}\label{thm_DemolitionHamiltonian}
		Let $\hbox{\input{modules/symbols/internalmeasurementSym.tex}}\!\!\!\!$ be a non-degenerate Hamiltonian with internal energy observable $\internalenergystructure$. Then the Hamiltonian and its composition with the counit $\raisebox{\SymFontShift}{\hbox{\input{modules/symbols/internaltimecounitSym.tex}}}$ are both classical injections. We shall refer to the latter as the \textbf{demolition Hamiltonian}.
	\end{theorem}
	\begin{proof}
		A Hamiltonian is non-degenerate, in the sense of def \ref{def_NondegenerateSystem}, if and only if all its eigenspaces are 1-dimensional. In the latter case, it is immediate to see that the desired maps are classical injections.
	\end{proof}

	\begin{theorem}\label{thm_InternalTime} \textbf{(Internal time observable)}\\
		Let $\hbox{\input{modules/symbols/internalalgebraSym.tex}}\!\!\!\!$ be a non-degenerate quantum dynamical system with time translation group $G \isom \integersMod{N}$, and some internal energy observable $\internalenergystructure$. A strongly complementary pair $(\internalenergymonoid,\internaltimecomonoid)$, with $\internaltimestructure$ an internal time observable, exists if and only if the demolition Hamiltonian $h$, seen as a classical injection $\internalenergycomonoid \monom (\raisebox{\SymFontShift}{\hbox{\input{modules/symbols/timecomultSym.tex}}}\!,\raisebox{\SymFontShift}{\hbox{\input{modules/symbols/timecounitSym.tex}}}\!)$, has image equal to a subgroup of the group $G^\wedge \isom \integersMod{N}$ of multiplicative characters. When it exists, such an internal time observable is unique and corresponds to a quotient group $G/K$; $\hbox{\input{modules/symbols/internalalgebraSym.tex}}\!\!\!\!$ is the left regular action of $G$ on $G/K$, translating external time states to internal time states; the demolition Hamiltonian translates internal energy levels (the characters of the quotient $G/K$) to external energy levels (characters of $G$).
	\end{theorem}
	\begin{proof} \textit{(sketch)}
		If an internal time observable exists and is strongly complementary to the internal energy structure, then it comes with its own group structure $G'$, and $h^\dagger$ (which is the group action $\hbox{\input{modules/symbols/internalalgebraSym.tex}}\!\!\!\!$ evaluated at the identity of $G'$) is a surjective group homomorphism $G \epim G'$, yielding $G' = G/K$ for $K$ the kernel; $h$ is then the corresponding injective group homomorphism $(G')^\wedge \monom G^\wedge$ between the Pontryagin duals. On the other hand, if $h$ has a subgroup $H$ of $G^\wedge$ as its image, then we can consider the unique $\dagger$-SCFA giving the copiables of $\internalenergystructure$ the group structure of $H$, making $h$ into an injective group homomorphism; then $h^\dagger$ is a surjective group homomorphism, giving a quotient $G/K$ as before.
	\end{proof}

	Finally, consider an external clock synchronised with a number of dynamical systems. 
	Now suppose the clock is discarded by fixing a total energy for the remaining system, and that one of these systems, call it $\hbox{\input{modules/symbols/internalalgebraSym.tex}}\!\!\!$, has non-degenerate Hamiltonian and an internal time observable. Then the following theorem shows that, with an appropriate initial state, the quantum dynamical system $\hbox{\input{modules/symbols/internalalgebraSym.tex}}\!\!\!\!$ acts as an internal clock for all the other systems. This is a key result, as it provides a mechanism explaining the emergence of quantum clocks in terms of bigger, unknown clocks which are only known through the total energy constraint. As before, multiple synchronised systems will be represented by a single system $\hbox{\input{modules/symbols/algebraSym.tex}}\!\!$, which is assumed to be appropriately separable.

	\begin{theorem}\label{thm_SynchronicityActionThm} 
		Let $\hbox{\input{modules/symbols/internalalgebraSym.tex}}\!\!\!\!:\SpaceG \tensor \timeobj \rightarrow \SpaceG$ be a quantum dynamical system with non-degenerate Hamiltonian, internal energy observable $\internalenergystructure$ and some internal time observable making $\SpaceG$ itself a potential clock. Assume $\hbox{\input{modules/symbols/internalalgebraSym.tex}}\!\!\!\!$ and $\hbox{\input{modules/symbols/algebraSym.tex}}\!\!\!\!:\SpaceH \tensor \timeobj \rightarrow \SpaceH$ are synchronised quantum dynamical systems with initial states $\raisebox{\SymFontShift}{\hbox{\input{modules/symbols/internaltimeunitSym.tex}}}$ and $\ket{\psi}$ respectively, the latter arbitrary, and with total energy $\chi$. Then the following map is a unitary dynamic governed by the clock $\SpaceG$:
		\begin{equation}\label{eqn_DynamicDescent}
			\hbox{\input{modules/pictures/DynamicDescent.tex}}
		\end{equation}
	\end{theorem}
	\begin{proof}
		\begin{equation}\label{eqn_DynamicDescentProof1}
			\hbox{\input{modules/pictures/DynamicDescentProof1.tex}}
		\end{equation}
		\begin{equation}\label{eqn_DynamicDescentProof2}
			\hbox{\input{modules/pictures/DynamicDescentProof2.tex}}
		\end{equation}
		and eqn \ref{eqn_DynamicsDef1} applied to $\hbox{\input{modules/symbols/algebraSym.tex}}\!\!\!\!$ completes the proof (observing that $\ket{\psi}$ is arbitrary).
	\end{proof}

	The ideas behind thm \ref{thm_SynchronicityActionThm} are also involved in explaining the emergence of finite-dimensional quantum clocks governing finite-dimensional quantum dynamical systems.

	\begin{remark}\label{rmrk_ContinuousSystems}
		Let $U(t)$ be a strongly continuous 1-parameter unitary group over some finite-dimensional system $\SpaceH$, and let $\nu_1,...,\nu_n$ be the frequencies corresponding to it. Then for any $\varepsilon>0$ there is a $T$ such that $T/\nu_1,...,T/\nu_n$ all differ from the nearmost integer by at most $\varepsilon$. If $\nu = \nu_j$ is the frequency maximal in absolute value, and $N$ is twice the integer part of $T/\nu$, then we can approximate the continuous quantum dynamical system to a precision of $\varepsilon$ with one governed by a finite-dimensional clock with time translation group $\integersMod{N}$. This corresponds to considering the separable action of the discrete subgroup $\frac{1}{2}\nu\integers$ on both the system and some internal clock corresponding to the group $\integersMod{N}$: in analogy with thm \ref{thm_SynchronicityActionThm}, we could see the (approximately) synchronised clock-system pair as obtained by taking the \inlineQuote{external clock} (ticking $\frac{1}{2}\nu\integers$) to be in the ground energy state.
	\end{remark}

\section{Stone's Theorem on 1-parameter unitary groups}
	\label{section_GeneratorsDynamics}

	The standard result in continuous time quantum mechanics that relates Hamiltonians to unitary dynamics is known as \textbf{Stone's theorem on 1-parameter unitary groups}, stating that the strongly continuous group homomorphisms $t \mapsto U_t$ (the unitary dynamic) from the additive reals $(\reals,+,0)$ to the unitary operators $\UnitaryOps{\SpaceH}$ over some Hilbert space $\SpaceH$ are exactly those in the form $U_t = \exp [i t H]$ for some (not necessarily bounded) self-adjoint operator $H$ on $\SpaceH$ (the \inlineQuote{traditional} Hamiltonian). In thm \ref{thm_DynamicsSpectra}, we saw that the relationship between dynamics and Hamiltonians is given, in our framework, by adjunction: this section will recast our results in the familiar form of Stone's Theorem.

	\begin{theorem}\textbf{Stone's Theorem}\label{thm_StoneThmOneParamUGroups}\\
		Let $t \mapsto U_t$ be a strongly continuous group homomorphism $\reals \rightarrow \UnitaryOps{\SpaceH}$, where $\SpaceH$ is any Hilbert space. Then there exists a unique self-adjoint operator $H: \SpaceH \rightarrow \SpaceH$, not necessarily bounded, such that $U_t = \exp[i t H ]$ for all $t : \reals$.
	\end{theorem}
	\begin{proof}
		See \cite{CQM-StoneOneParameter}.
	\end{proof}

	\begin{theorem}\textbf{Spectral Theorem}\label{thm_SpectralTheorem}\\
		Let $H: \SpaceH \rightarrow \SpaceH$ be a self-adjoint operator. Then there is a measurable space $X$, a measure $\mu$ and a unitary isomorphism $V: \SpaceH \rightarrow \Ltwo{X,\mu}$ such that $V H V^\dagger = H'$ is a multiplication operator:
		\begin{align*}
			H' : 	\Ltwo{X,\mu} & \rightarrow 	\Ltwo{X,\mu}\\
					\psi		 & \mapsto 	 	(x \mapsto E_x \psi(x)) \numberthis \label{eqn_SpectralTheorem} 
		\end{align*}
		We will refer to the measurable function $E: X \rightarrow \reals$ as the \textbf{spectrum} of the operator $H$. If $H$ is bounded then $E$ is essentially bounded and we have $\Norm{}{H'} = \Norm{\infty}{E}$.
	\end{theorem}
	\begin{proof}
		See \cite{hall2013quantum}.
	\end{proof}

	This is the usual way to derive energy levels for a quantum dynamical system: unfortunately, it turns out not to be canonical. This may seem a merely categorical flaw, but it is in fact related to an important physical fact: valuing energy in the reals is necessarily subject to a choice of units of measurement.\\

	We have seen in section \ref{section_EnergyObservable} that a canonical space for the energy levels associated with a unitary dynamic (governed by some abelian group $G$) is given by the Pontryagin dual $G^\wedge$: any attempt to faithfully value energy in some other space would be equivalent to a choice of group isomorphism $G^\wedge \isom V$ for some $V$. Similarly, when the dynamics are governed by $G = \reals$ we expect any valuation of then energy in $V = \reals$ to be conditional on some choice of units of measurement, i.e. on fixing some isomorphism $\reals^\wedge \isom \reals$.\\

	Units of measurement, seen as group isomorphisms $G^\wedge \stackrel{\isom}{\rightarrow} V$, form a homogeneous space under (transitive and faithful) left regular action of the group automorphisms of $V$. The action corresponds to changing units, and for $V = \reals$ it is the usual multiplication by some non-zero real number. Thus the Hamiltonian and energy spectrum obtained from thms \ref{thm_StoneThmOneParamUGroups} and \ref{thm_SpectralTheorem} are subject to an underlying choice of units of measure $\reals^\wedge \stackrel{\isom}{\longrightarrow} \reals$: thm \ref{thm_StoneThmOneParamUGroups2} will make this statement precise.

	\begin{theorem}\label{thm_StoneThmOneParamUGroups2}
		The isomorphisms $\Isoms{\AbCategory}{\reals^\wedge}{\reals}$ form a homogeneous space under the (faithful and transitive) left regular action of $\Automs{\AbCategory}{\reals}$. Also $\Automs{\AbCategory}{\reals} \isom_{\AbCategory} (\reals^\times,\cdot,1)$, where $\alpha_c := x \mapsto c \cdot x $ is the automorphism corresponding to a non-zero real $c$.
		As a consequence, the bijection of thm \ref{thm_StoneThmOneParamUGroups} is non-canonical, and there is instead a homogeneous space of bijections $U_t = \exp[i t \frac{1}{\hbar} H ]$ between strongly continuous group homomorphisms $(U_t)_{t:\reals}$ and self-adjoint operators $H$, with fiber isomorphic to the homogeneous space $\Isoms{\AbCategory}{\reals^\wedge}{\reals}$ (except at the singular point $(U_t)_t = (\id{\SpaceH})_t$). Singling out one such bijection is equivalent to fixing a choice of isomorphism $\reals^\wedge \isom \reals$.
	\end{theorem}
	\begin{proof}
		The first two observations are standard checks. To see that the bijections form a homogeneous space, all we have to show is that there is an action of $\Automs{\AbCategory}{\reals}$ on them: the action of a $\frac{h}{h'} : \reals^\times$ on the space of bijections is given as follows: 
		\begin{equation}\label{eqn_StoneThmv2HomogeneousSpaceBijs}
			\frac{h}{h'} : U_t = \exp[i t \frac{1}{\hbar} H ] \mapsto U_t = \exp[i t \frac{1}{\hbar'} H' ]
		\end{equation}
	\end{proof}

	Taking thms \ref{thm_StoneThmOneParamUGroups} and \ref{thm_SpectralTheorem} together, the \textbf{energy spectrum} for a unitary dynamic $(U_t)_t$ is usually defined to be $E: X \rightarrow \reals$. However, it is a consequence of thm \ref{thm_StoneThmOneParamUGroups2} that this energy spectrum is non-canonical, depending on a particular choice of unit of measurement for the energy: we will denote by $E^\hbar$ the spectrum associated with a particular bijection $U_t = \exp[i t \frac{1}{\hbar} H ]$. We can, however, define a canonical energy spectrum $\hat{E}: X \rightarrow \reals^\wedge$.

	\begin{theorem}\textbf{Canonical energy spectrum}\label{thm_CanonicalSpectrum}\\
		Let $(U_t)_{t:\reals}$ be a strongly continuous group homomorphism $\reals \rightarrow \UnitaryOps{\SpaceH}$. Fix a bijection $U_t = \exp[i t \frac{1}{\hbar} H ]$, and obtain the\footnote{The decomposition is not really unique. However, the same $X$ works for all $\hbar$, and the construction of $E^\hbar$ is contravariantly functorial with respect to the choice of $X$. So we shall not worry about this any further.} spectral decomposition with $V: \SpaceH \rightarrow \Ltwo{X,\mu}$ and $E^\hbar: X \rightarrow \reals$. Define $\hat{E} : \Ltwo{X,\mu} \rightarrow \reals^\wedge$ by:
		\begin{equation}\label{eqn_CanonicalEnergySpectrumDef}
			x  \mapsto (t \mapsto \exp[i \frac{1}{\hbar} E^\hbar_x t])
		\end{equation}
		Then $\hat{E}$ is independent of the choice of $\hbar : \reals^\times$ (i.e. it is canonical) and we shall refer to it is as the \textbf{canonical energy spectrum} of $(U_t)_t$.
	\end{theorem}
	\begin{proof}
		The action defined in eqn \ref{eqn_StoneThmv2HomogeneousSpaceBijs} sends $E^\hbar$ to $E^{\hbar'} = \frac{\hbar'}{\hbar} E^\hbar$. Thus eqn \ref{eqn_CanonicalEnergySpectrumDef} is invariant under action of $\frac{\hbar}{\hbar'} : \reals^\times$.
	\end{proof}

	\begin{remark}\textbf{(Non-demolition Hamiltonian?)}\\
		Given a dynamic $(U_t)_{t:\reals}$ and its canonical energy spectrum $\hat{E}$, we can \inlineQuote{construct} an operator $\hat{H}: \Ltwo{X,\mu} \rightarrow \Ltwo{X,\mu} \tensor \Ltwo{\reals^\wedge}^\star$ similar to the one of thm \ref{thm_DynamicsSpectra} by using delta functions:
		\begin{equation} \label{eqn_DynamicsSpectraReprise}
			\hat{H}: \int_X a_x \ket{x} d\mu(x) \mapsto \int_X a_x \ket{x}\tensor \ket{\hat{E}_x} d\mu(x)
		\end{equation}
	 	We denoted by $\ket{x}$ the delta function at $x : X$ and by $\ket{\hat{E}_x}$ the delta function at $\hat{E}_x : \reals^\wedge$. Subject to a choice $f: \reals^\wedge \isom \reals$ of units of measurement, we can also recover the Hamiltonian of thm \ref{thm_StoneThmOneParamUGroups} from the non-demolition Hamiltonian \inlineQuote{constructed} above:
		\begin{equation} \label{eqn_TraditionalHamiltonian}
			VHV^\dagger = \left(\id{\Ltwo{X,\mu}} \tensor \int_{\reals^\wedge} f(\chi) \bra{\chi} d\chi\right) \cdot \hat{H}
		\end{equation}
		The operator $\int_{\reals^\wedge} f(\chi) \bra{\chi} d\chi$ is nothing but $f$ extended linearly on the basis of delta functions for $\reals^\wedge$.
	\end{remark}

	The non-demolition Hamiltonian above is not fully rigorous, so we need take a different road to link Stone's Theorem with our finite group dynamics. Recasting the results in terms of projection-valued measures provides us with a viable alternative.

	\begin{lemma}\label{lemma_ProjectionValuedSpectrum}
		Let $X,Y$ be measurable spaces (with sigma-algebras $\Sigma_X$ and $\Sigma_Y$), $\mu$ a measure on $X$ and $f: X \rightarrow Y$ measurable. Then $f$ determines a projection-valued measure $\pi_f: \Sigma_Y \rightarrow \Bounded{\Ltwo{X,\mu}}$ by:
		\begin{equation}\label{eqn_ProjectionValuedSpectrum}
			\pi_f(U) = \text{projection onto subspace } \Ltwo{f^{-1}(U),\mu}
		\end{equation}
		for all $U : \Sigma_Y$. If $V: \SpaceH \rightarrow \Ltwo{X,\mu}$ is a unitary, then $\pi_E$ can be seen (giving $V$ as understood) as a projection valued measure $\Sigma_Y \rightarrow \Bounded{\SpaceH}$ by considering $V^\dagger \pi_E V$.
	\end{lemma}

	\begin{theorem}\textbf{Spectral Theorem (projection-valued)}\label{thm_SpectralTheoremProj}\\
		Let $H: \SpaceH \rightarrow \SpaceH$ be a self-adjoint operator. Let $V: \SpaceH \rightarrow \Ltwo{X,\mu}$ and spectrum $E: X \rightarrow \reals$ be given by thm \ref{thm_SpectralTheorem}. If $\pi_E$ is the projection-valued measure defined by lemma \ref{lemma_ProjectionValuedSpectrum}, 
		\newpage
		then we can reconstruct $H$ as:
		\begin{equation}
			H = \int_\reals \! \lambda \, d\pi_E(\lambda)
		\end{equation}
	\end{theorem}

	\begin{theorem}\textbf{Stone's Theorem (projection-valued)}\label{thm_StoneThmOneParamUGroupsProjValued}\\
		Let $(U_t)_{t:\reals}$ be a strongly continuous group homomorphism $\reals \rightarrow \UnitaryOps{\SpaceH}$. Let $V: \SpaceH \rightarrow \Ltwo{X,\mu}$ unitary isomorphism and $\hat{E}: X \rightarrow \reals^\wedge$ canonical energy spectrum be given by thm \ref{thm_CanonicalSpectrum}. If $\pi_{\hat{E}}$ is the projection-valued measure defined by lemma \ref{lemma_ProjectionValuedSpectrum}, then we can reconstruct $(U_t)_t$ as:
		\begin{equation}
			U_t = \int_{\reals^\wedge} \!\! \chi(t) \, d\pi_{\hat{E}}(\chi)
		\end{equation}		
	\end{theorem}

	Finally, the form of Stone's theorem on 1-parameter unitary groups given by thm \ref{thm_StoneThmOneParamUGroupsProjValued} can be extended to the dynamics described in this work, remembering that the dynamics of def \ref{def_Dynamics} correspond exactly to (necessarily strongly continuous) group homomorphisms $G \rightarrow \UnitaryOps{\SpaceH}$, where $G$ is a finite abelian group and $\SpaceH$ is finite-dimensional.

	\begin{theorem}\textbf{Canonical energy spectrum (finite groups)}\label{thm_CanonicalSpectrumFinite}\\
		Let $(U_t)_{t:G}$ be a group homomorphism $G \rightarrow \UnitaryOps{\SpaceH}$. Let $\hat{H}: \SpaceH \rightarrow \SpaceH \tensor \Ltwo{G^\wedge}^\star$ be the corresponding Hamiltonian from thm \ref{thm_DynamicsSpectra}, and $X$  an orthonormal basis of eigenvalues for $\hat{H}$. Let $V: \SpaceH \rightarrow \Ltwo{X}$ be the unitary corresponding to the basis, and define the canonical energy spectrum $\hat{E}:X \rightarrow G^\wedge$ via the character basis of $\Ltwo{G^\wedge}$:
		\begin{equation}
			\hat{E}(\ket{x}) := \left[\left(\bra{x} \tensor \id{\Ltwo{G^\wedge}^\star}\right)\hat{H} \ket{x}\right]^\dagger
		\end{equation} 
		Then the projection-valued measure $\pi_{\hat{E}}$ is independent of the choice of basis\footnote{Seen as having projections in $\Bounded{\SpaceH}$, its correct form would be $V^\dagger \pi_{\hat{E}} V$, where $E$ is dependent on the choice of $V$. The statement here is that the entire expression $V^\dagger \pi_{\hat{E}} V$ is independent of the choice of $V$.} and coincides with the complete family of orthogonal projectors defined by $\hat{H}$.
	\end{theorem}

	\begin{theorem}\textbf{Stone's Theorem (finite groups)}\label{thm_StoneThmFinite}\\
		Let $(U_t)_{t:G}$ be a group homomorphism $G \rightarrow \UnitaryOps{\SpaceH}$, with $V: \SpaceH \rightarrow \Ltwo{X}$ and canonical energy spectrum $\hat{E}:X \rightarrow \reals^\wedge$ given by thm \ref{thm_CanonicalSpectrumFinite}. If $\pi_{\hat{E}}$ is the projection-valued measure defined by lemma \ref{lemma_ProjectionValuedSpectrum}, then we can reconstruct $(U_t)_t$ as:
		\begin{equation}\label{eqn_DynamicIntegralProjectors}
			U_t = \int_{G^\wedge} \!\! \chi(t) \, d\pi_{\hat{E}}(\chi)
		\end{equation} 		
	\end{theorem}

	The measure provided by thm \ref{thm_CanonicalSpectrumFinite} can be extended linearly to obtain the Hamiltonian $\hat{H}$ in the form of def \ref{def_Spectra}, and similar extension turns $(U_t)_{t:G}$ into a dynamic in the form of def \ref{def_Dynamics}. Under these circumstances, eqn \ref{eqn_DynamicIntegralProjectors} yields the same statement as thm \ref{thm_DynamicsSpectra}, closing the link between Stone's Theorem on 1-parameter unitary groups and the dynamics/Hamiltonians duality presented previously.

\section{General dynamics}
	\label{section_GeneralDynamics}

	The topics in this paper have been presented from the point of view of the cyclic groups $\integersMod{N}$ in order to exploit the parallelisms with continuous time, but the entire framework 
	\newpage
	fits well into the more general understanding of time translation in quantum dynamical systems as the action of a certain symmetry group. In this more general context, presented more in detail in \cite{StefanoGogioso-RepTheoryCQM}, one works with (unitary) representations of symmetry groups encoded by strongly complementary pairs $\dagger$-FA/$\dagger$-SCFA in exactly the same way as done here  with unitary dynamics; the Hamiltonian gives the conserved quantities for the symmetry in the form of the irreducible characters (not just the multiplicative ones) of the (possibly non-abelian) group. Almost all constructions presented here readily generalise to the case of abelian symmetry groups; the generalisation to non-abelian groups, however, is tricky and will be deferred to future work. 

\section{Conclusions and future work}
	\label{section_ConclusionsFutureWork}

	We have provided the key definitions and results needed for a systematic treatment of quantum dynamics in the framework of CQM. We have defined quantum clocks and shown their relation to strongly complementary structures. We have highlighted the duality between unitary dynamics and observables, and used it to define Hamiltonians and Fourier transforms. We have proven the equivalence of strong complementarity and canonical commutation of observables, and concluded the necessary existence of a conjugate pair of time/energy observables for quantum clocks, satisfying an adequate uncertainty principle. We have given a monadic definition of Schr\"{o}dinger's equation, and used it to show the validity of the Feynman's clock construction in our framework. Touching base with a more concrete understanding of clocks and dynamical systems, we have defined families of systems synchronised with a clock, and shown that \inlineQuote{forgetting} the clock is equivalent to setting a conserved total energy for the remaining systems. We have also proven necessary and sufficient conditions for the existence of an internal time observable, and used the latter to explain how and when a quantum dynamical system in a synchronised family can behave as a clock governing the other systems in the family. Finally, we have proven a finite-dimensional analogue to Stone's Theorem on 1-parameter unitary groups that applies to our case.\\
	
	As explained briefly in section \ref{section_GeneralDynamics}, this work on time and energy can be framed in the more general context of symmetries of quantum system and their associated conserved quantities: CQM already contains most of the tools needed (collected in \cite{StefanoGogioso-RepTheoryCQM}) for a treatment of the problem in terms of representation theory, but the generalisation of a number of results to the non-abelian case turned out to be especially tricky. It will certainly be interesting to expand the work in this direction in the future. Furthermore, remark \ref{rmrk_ContinuousSystems} applies the ideas of thm \ref{thm_SynchronicityActionThm} to the case of the time translation group $\reals$, but it does so informally. An infinite-dimensional generalisation of the methods, e.g. to $\integers$ and $\reals$, would also be of major interest: unfortunately some of the key structures used in this work exist only in finite dimensions, and the correct infinite-dimensional generalisation is an active area of research in CQM.

\section*{Acknowledgements}
	The author would like to thank Bob Coecke, Aleks Kissinger, Clare Horsman, Amar Hadzihasanovic, Will Zeng, Sukrita Chatterji and Nicol\`o Chiappori for plenty of comments, suggestions, useful discussions and support. Funding from EPSRC and Trinity College is gratefully acknowledged. 

\bibliographystyle{IEEEtran}
\bibliography{bibliography/CategoryTheory,bibliography/CategoricalQM,bibliography/NonLocalityContextuality,bibliography/QuantumComputing,bibliography/ClassicalMechanics,bibliography/LogicComputation,bibliography/Gravitation,bibliography/QFT,bibliography/StatisticalPhysics,bibliography/Misc,bibliography/StefanoGogioso}

\end{document}